\documentclass[11pt,authoryear]{elsarticle}
\usepackage[margin=2.3cm, includehead]{geometry}
\usepackage{bm}
\makeatletter
\def\ps@pprintTitle{%
  \let\@oddhead\@empty
  \let\@evenhead\@empty
  \let\@oddfoot\@empty
  \let\@evenfoot\@oddfoot
}
\makeatother
\usepackage{currfile}
\usepackage{mathrsfs}
\usepackage[mathscr]{eucal}
\DeclareMathAlphabet{\mathpzc}{OT1}{pzc}{m}{it}
\usepackage{natbib}
\usepackage{makeidx}
\usepackage{multirow}
\usepackage{multicol}
\usepackage{booktabs}
\usepackage{enumerate}
\usepackage{setspace}
\usepackage{float}
\usepackage{hhline}
\usepackage{bbm}
\usepackage[dvipsnames,svgnames,table]{xcolor}
\usepackage{graphicx}
\usepackage{epstopdf}
\usepackage[normalem]{ulem}
\usepackage{listings}
\usepackage{kotex}
\usepackage{amsmath}
\usepackage{amsfonts}
\usepackage{url}
\usepackage{amssymb}
\usepackage{dsfont}

\usepackage{epsfig}
\usepackage{verbatim}
\usepackage{amsthm}
\usepackage{framed}    % to make framed boxes
\usepackage{scalefnt}
\usepackage{alltt}
\usepackage{tikz}
\usetikzlibrary{arrows.meta}
\usepackage{hyperref}
\hypersetup{
    bookmarks=true,         % show bookmarks bar?
    unicode=false,          % non-Latin characters in Acrobat's bookmarks
    pdftoolbar=true,        % show Acrobat's toolbar?
    pdfmenubar=true,        % show Acrobat's menu?
    pdffitwindow=true,      % page fit to window when opened
    pdftitle={My title},    % title
    pdfauthor={Author},     % author
    pdfsubject={Subject},   % subject of the document
    pdfnewwindow=true,      % links in new window
    pdfkeywords={keywords}, % list of keywords
    colorlinks=true,        % false: boxed links; true: colored links
    linkcolor=blue,         % color of internal links
    citecolor=blue,        % color of links to bibliography
    filecolor=magenta,      % color of file links
    urlcolor=cyan           % color of external links
}

\def\eqd{\buildrel {\rm (d)} \over =}

\def\cw#1 { \overset{\mathbb{P}}{\underset{#1}{\longrightarrow}} }
\def\Real{\mathbb{R}}
\def\Natu0{\mathbb{N}_0}

\def\E#1{{\mathbb E}\left[#1\right]}

\def\Var#1{{\rm Var}\left(#1\right)}

\def \rcov#1#2 {{\rm cov}_{#1}\left( #2\right)}

\def\cov#1{{\rm  Cov}\left(#1\right)}

\newtheorem{example}{Example}

\newtheorem{lemma}{Lemma}
\newtheorem{theorem}{Theorem}

\newtheorem{corollary}{Corollary}
\newtheorem{remark}{Remark}
\newtheorem{proposition}{Proposition}

\newtheorem{model}{Model}

\makeatletter
\newtheoremstyle{nopunct}% name
  {}{}              % Space above/below
  {\itshape}        % Body font
  {}                % Indent amount
  {\bfseries}       % Theorem head font
  {}                % Punctuation after theorem head
  { }               % Space after theorem head
  {}                % Theorem head spec
\makeatother

\theoremstyle{nopunct}

\numberwithin{equation}{section}

\begin{document}
	\begin{titlepage}	
		\thispagestyle{empty}
%\begin{frontmatter}
\title{State-Space Representations of Negative-Binomial INGARCH Models for Insurance Ratemaking}
 %(\currfilename)}

\author{Jae Youn Ahn\fnref{thirdfoot}}%\corref{bbb}}
\author{\qquad Hong Beng Lim \fnref{fifthfoot}}
\author{\qquad Mario V.~W\"uthrich\fnref{fourthfoot}}%\corref{bbb}}

%\cortext[bbb]{Corresponding authors/equal contribution}
\fntext[thirdfoot]{Department of Statistics, Ewha Womans University, Seoul, Republic of Korea. Email: \url{jaeyahn@ewha.ac.kr}}
\fntext[fifthfoot]{Department of Finance, Chinese University of Hong Kong, Hong Kong, China. Email: \url{benhblim@cuhk.edu.hk}}
\fntext[fourthfoot]{Department of Mathematics, ETH Zurich, Switzerland.
Email: \url{mario.wuethrich@math.ethz.ch}}

\begin{abstract}

Integer-valued generalized autoregressive conditional heteroskedastic (INGARCH) models are a popular framework for modeling serial dependence in count time series, with the dynamics specified through a recursion for the conditional mean.
While convenient for modeling, prediction, and estimation, their standard formulation can be restrictive in insurance applications, where explanatory variables are essential for capturing policyholder heterogeneity and a suitable notion of stationarity is particularly useful for ratemaking, given the short time series typically available in insurance panels.
These requirements are difficult to reconcile in the classical INGARCH formulation because the conditional mean plays two roles simultaneously: it governs the temporal evolution of the process and determines the scale of the observations.
To address these issues, we propose a heterogeneous negative-binomial INGARCH(1,1) model for insurance ratemaking, in which the predictive mean is decomposed into a time-varying a priori rate and a baseline conditional mean.
The a priori rate accommodates time-varying heterogeneity, while the baseline conditional mean preserves the affine INGARCH recursion.
Since time-varying heterogeneity makes stationarity of the observed count process inappropriate, weak stationarity is instead formulated at the level of the latent state process through an equivalent state-space model representation.
The practical implementation of the model is illustrated using panel insurance data, highlighting its applicability to predictive analysis in property and casualty insurance ratemaking.

\end{abstract}

\maketitle


\vspace{.2cm}


\vspace{.2cm}

Version of \today
\end{titlepage}

\newpage

\section{Introduction}

\subsection{INGARCH model}

Claims counts modeling is an important topic in insurance pricing and risk management. 
While random-effects models \citep{nelder1997credibility} are central for capturing policyholder-level variation across time and deriving fair premiums, it has been increasingly recognized that accounting for temporal dependence is equally essential for capturing the dynamic structure of claims histories \citep{pinquet2001allowance, bolance2003time, gourieroux2004heterogeneous}.
Alongside the integer-valued autoregressive (INAR) model \citep{al1987first}, 
the \emph{integer-valued generalized autoregressive conditional heteroskedastic} (INGARCH) model 
has emerged as a popular framework for representing such temporal dependence in counts processes.
%The INGARCH model is a flexible framework for modeling discrete-valued time-series with serial dependence.
Introduced by \citet{ferland2006integer} and extended by \citet{fokianos2009poisson}, it combines an autoregressive recursion for the conditional mean with a specified conditional distribution, often Poisson or
negative-binomial (NB) for claims counts.
In particular, the INGARCH models were inspired by analogous specifications for continuous-valued data, including the generalized autoregressive conditional heteroskedastic (GARCH) model for Gaussian observations \citep{bollerslev1986generalized}, the autoregressive conditional duration (ACD) model for Exponential or Gamma-distributed durations \citep{engle1998autoregressive}, and related Gamma autoregressive models \citep{jasiak1999persistence}.
These parallels indicate that the INGARCH approach belongs to a broader class of conditional mean models applicable across different data types.

In the INGARCH setting, the one-step-ahead conditional mean is a linear combination of a constant term, the latest observation(s), and the preceding conditional mean(s), allowing the model to capture a wide range of temporal patterns in counts data.
Its parameters can be estimated directly via the conditional log-likelihood, avoiding the need for simulation-based procedures.
This combination of modeling flexibility and computational simplicity has led to applications in various areas, including climate science \citep{jahn2023approximately}, actuarial science \citep{reboredo2023forecasting}, epidemiology \citep{fokianos2010interventions}, and many other domains where count data exhibits dynamic dependence.
Well-known special cases are the Poisson-INGARCH model \citep{ferland2006integer} and its NB extensions \citep{zhu2011negative, xu2012model}. 

When applying INGARCH models in insurance practice, several issues can arise.
First, bringing exogenous covariates into an INGARCH specification can be delicate, particularly due to the requirement that the conditional mean remain strictly positive. 
%Second, because the INGARCH recursion depends directly on the past observations, any missing value breaks the recursion. 
Second, INGARCH models are often motivated through long-run stability or stationarity arguments, whereas insurance applications are typically characterized by short finite histories, often no more than 10 observation periods. This motivates a formulation which by design satisfies a suitable notion of stationarity.

\subsection{State-space models}\label{sec.ssm}

\emph{State-space models} (SSMs) are a broad class of dynamic models in which
an evolving latent state is used to describe the observed outcome process
\citep{durbin2012time}. In insurance ratemaking, the latent state can be
interpreted as an unobserved risk level, while the observed process corresponds
to the claim counts. Following \citet{cox1981statistical}, SSMs are often
classified as parameter-driven or observation-driven, depending on whether the
transition law of the latent state is independent of, or explicitly depends on,
the observation history conditional on the previous latent state. Although we do
not rely on a sharp distinction between these two classes, the models considered
in this paper are observation-driven in this broad sense, because the evolution
of the latent state may depend on both past observations and the previous latent
state. 

A state-space formulation separates the model into two components. The
observation equation specifies the conditional distribution of the observed
count given the current latent state, while the evolution equation describes how
the latent state is updated over time. This separation provides a natural way to
distinguish the time-varying scale of the observations from the temporal
evolution of the latent risk level. In particular, explanatory variables can be
incorporated through the observation layer, whereas serial dependence can be
modeled through the latent-state dynamics. Consequently, if an INGARCH model
admits an equivalent state-space representation, the practical difficulties
discussed above can be addressed within a unified modeling framework.

\subsection{From the INGARCH model to the SSM}

To connect the INGARCH model with a state-space formulation, we use the
notion of a \emph{marginalized state-space model} (M-SSM)\footnote{See
Model~\ref{mod.2} in Section~\ref{sec.4} for a formal definition of M-SSM.}
\citep{brockwell2002introduction}. As in a standard SSM, an M-SSM is
specified through an observation equation and an evolution equation. The key
difference is that the evolution equation specifies only the conditional
marginal distribution of the next latent state, rather than the full joint
evolution of successive latent states. 

Using this framework, we show that the NB-INGARCH model admits an equivalent
M-SSM representation with a Poisson observation equation and a Gamma marginal
evolution equation for the latent state. This equivalence allows us to address
two practical issues that arise when INGARCH models are used for insurance
ratemaking.

First, time-varying covariates can be incorporated naturally through the
a priori rate \(\lambda_t\) in the observation layer of the M-SSM. Through the
equivalence between the two representations, this covariate structure is then
inherited by the corresponding NB-INGARCH model.

Second, stationarity under time-varying heterogeneity can be formulated at the
level of the latent states, after separating out the time-varying scale of the
observations. For this purpose, we introduce a particular M-SSM specification
in which the latent states have constant mean and variance over time. However,
because the M-SSM specifies only conditional marginal distributions, it does not
identify the conditional joint distribution of successive latent states, and
therefore does not determine the covariance structure of the latent process.
Consequently, weak stationarity cannot be fully established within the M-SSM
formulation alone.

To resolve this issue, we show that the M-SSM corresponding to the NB-INGARCH
model can be extended, in a compatible way, to an SSM in which the joint
evolution of successive latent states is specified.
Although such an extension need not be unique, we show that, once equipped with the additional property that the conditional mean of the latent state evolves linearly---analogously to the NB-INGARCH model---the covariance structure of the
latent process is uniquely determined for all compatible SSM representations.
In particular, under a further restriction on the recursion parameters, the latent state
process is weakly stationary. Thus, the proposed M-SSM, and equivalently the
NB-INGARCH model, incorporates time-varying heterogeneity while
admitting a weakly stationary latent-state interpretation.

The remainder of the paper is organized as follows.
Section~\ref{sec.2} introduces the INGARCH model and its limitations in insurance ratemaking.
Section~\ref{sec.4}, which contains the main contribution of this paper, develops the NB-INGARCH model with time-varying a priori rates and introduces its corresponding M-SSM formulation.
Using this alternative representation, we show how covariates can be incorporated systematically, while stationarity is formulated at the level of the latent states.
Section~\ref{sec.5} provides the theoretical support for the weak-stationarity interpretation used in Section~\ref{sec.4}.
This section is rather technical and may be skipped by readers mainly interested in the modeling framework and its practical implementation. 
Section~\ref{sec.6} illustrates the practical implementation of the model based on a real data analysis, and the final section concludes the paper.

\section{The INGARCH model}
\label{sec.2}

The main class of models we consider is the INGARCH model. We first introduce this class of models, with a particular focus on the negative-binomial INGARCH (NB-INGARCH) model. We then discuss their limitations in incorporating covariates. 

\subsection{The INGARCH model and the NB-INGARCH}

We begin by introducing the notation that is used throughout this paper. We generically use the notation
$Z_{\le t}=(Z_s)_{s \le t}$, being all components of the doubly-infinite sequence
$(Z_s)_{s\in\mathbb{Z}}$ before time $t$, moreover, we write $Z_{1:t}=(Z_1, \ldots, Z_t)$. 
We use the following parametric families:
\begin{itemize}
\item $\mathrm{Pois}(\lambda)$: Poisson distribution with mean $\lambda\geq 0$. We adopt the convention that $\mathrm{Pois}(0)$ denotes the degenerate distribution at $0$, that is, if $X\sim \mathrm{Pois}(0)$, then $X=0$ almost surely.
  \item $\mathrm{Gamma}(\alpha,\beta)$: Gamma distribution with shape $\alpha>0$ and rate $\beta>0$ parameters, i.e., having mean ${\alpha}/{\beta}$ and variance ${\alpha}/{\beta^2}$.
  \item $\mathrm{NB}(\kappa,\pi)$: Negative-binomial distribution with $\kappa>0$, $\pi\in[0,1)$, mean ${\kappa\pi}/{(1-\pi)}$ and variance ${\kappa\pi}/{(1-\pi)^2}$.
        We adopt the convention that $\mathrm{NB}(\kappa,0)$ denotes the degenerate distribution at $0$, that is, if $X\sim \mathrm{NB}(\kappa,0)$, then $X=0$ almost surely.
       
\end{itemize}

For brevity of notation, we begin by considering the case of a single time series: the case of panel data, i.e., cross-sectional time series will be considered later in Section \ref{sec.6}. Thus, our main objective is to model the integer-valued time series $(Z_t)_{t\in\mathbb{Z}}$, where $Z_t \in \mathbb{N}_0$.

The INGARCH$(p,q)$ model was formalized by \citet{fokianos2009poisson}, initially with a Poisson distribution for the counting response. In this model, the conditional mean evolves as a linear combination of a constant term, the \(p\) most recent observations, and the \(q\) most recent conditional means. In this paper, we restrict attention to the first-order INGARCH$(1,1)$ specification, in which only the latest observation and the immediately preceding conditional mean enter the recursion. For brevity, we omit the order notation and refer to the INGARCH$(1,1)$ model simply as the INGARCH model.
More precisely, let \(\beta_t^{[0]} \geq 0\), and let \(\beta_t^{[1]}, \beta_t^{[2]} >0\). Under this first-order Poisson specification, it is assumed that
\[
Z_{t+1}\mid Z_{\leq t} \sim \mathrm{Pois}(M_{t+1\mid t}),
\]
where the conditional mean \(M_{t+1\mid t}\) satisfies the recursion
\[
M_{t+1\mid t}
=
\beta_t^{[0]}+\beta_t^{[1]}Z_t+\beta_t^{[2]}M_{t\mid t-1}.
\]
Equidispersion under the Poisson model means that the above recursion also defines a recursion relation for the variance, thus, connecting this model to the standard GARCH(1,1) model where the conditional variance evolves as a linear combination of past lags of the observed data and the conditional variances.

However, this same equidispersion property makes the Poisson-INGARCH model ill-suited to insurance data, which often exhibit overdispersion. Therefore, the primary focus of this paper is on the NB-INGARCH model of \citet{zhu2011negative} and \citet{xu2012model}, which accommodates overdispersion. Below, we provide a generalized formulation of this model allowing for time-varying coefficients.

\begin{model}[NB-INGARCH model]\label{mod.1}
Given the exogenous sequences
\begin{equation}\label{eq.q1}
(\beta_t^{[0]})_{t\in\mathbb{Z}}\subseteq [0,\infty), 
\qquad 
(\beta_t^{[1]})_{t\in\mathbb{Z}}\subseteq (0,\infty),
\qquad 
(\beta_t^{[2]})_{t\in\mathbb{Z}}\subseteq (0,\infty),
\end{equation}
we define the NB-INGARCH model for $(Z_t)_{t\in\mathbb{Z}}$ by the following specifications:
\begin{itemize}
  \item[i.] The conditional mean
  \[
  M_{t+1\mid t}:=\E{Z_{t+1}\mid Z_{\le t}}
  \]
  satisfies the recursion
  \begin{equation}\label{eq.nb1}
  M_{t+1\mid t}
  =
  \beta_t^{[0]}+\beta_t^{[1]}Z_t+\beta_t^{[2]}M_{t\mid t-1}.
  \end{equation}

  \item[ii.] Conditionally on \(Z_{\le t}\),
  \[
  Z_{t+1}\mid Z_{\le t}\sim \operatorname{NB}\bigl(\kappa_{t+1\mid t},\pi_{t+1\mid t}\bigr),
  \]
  with both of $\kappa_{t+1\mid t}>0$ and $\pi_{t+1\mid t}\in(0,1)$ being known functions of \(Z_{\le t}\) as well as of the exogenous sequences,
  with this dependence suppressed in the notation.

\end{itemize}
\end{model}

Note that the NB-INGARCH$(1,1)$ specification above is not uniquely determined by the mean recursion in \eqref{eq.nb1}. 
Indeed, the pair $(\kappa_{t+1\mid t},\pi_{t+1\mid t})$ is only constrained through the identity
\[
\frac{\kappa_{t+1\mid t}\,\pi_{t+1\mid t}}{1-\pi_{t+1\mid t}} \;=\; M_{t+1\mid t},
\]
satisfying \eqref{eq.nb1}, 
so many choices of $(\kappa_{t+1\mid t},\pi_{t+1\mid t})$ can generate the same conditional mean.
To obtain a fully specified NB-INGARCH$(1,1)$ model, one must therefore either fix one of the two parameters by imposing a structure on $\kappa_{t+1\mid t}$ or $\pi_{t+1\mid t}$, or introduce an additional restriction that links $\kappa_{t+1\mid t}$ and $\pi_{t+1\mid t}$. We refer to \citet{zhu2011negative} and \citet{xu2012model} for examples of such restrictions on $\kappa_{t+1\mid t}$ or $\pi_{t+1\mid t}$.

\subsection{Limitations of the INGARCH model as a ratemaking tool in property and casualty insurance}\label{sec.2.2}

In insurance applications, count time series arise naturally, for example, in property and casualty insurance. 
A key appeal of INGARCH models in this setting is that the dynamics are specified directly through a recursion for the conditional mean, which leads to a tractable likelihood function in an autoregressive form and a straightforward estimation procedure. 
Despite these advantages, however, INGARCH models have received relatively limited attention in the actuarial literature \citep{afazali2026modeling}. 
The limitations of INGARCH models discussed below are particularly relevant for property and casualty insurance ratemaking, and we examine them through that practical lens.

The first issue relates to incorporating covariates. The modeler (actuary) may wish to adjust the conditional mean with respect to policyholder-specific and time-varying covariates. 
Adding covariates as additional linear terms to the recursion \eqref{eq.nb1} requires that the additional terms remain strictly non-negative, which typically requires the same condition to be imposed on both the coefficients and the regressors \citep{chen2017bayesian, chen2019markov}. Such constraints may hinder the use of certain regressors, or preclude negative relationships between the conditional mean and the covariates. Indeed, one may sidestep such constraints by instead using a log-link, but the dynamics then operate on the log-means; thus, the conditional means are no longer an affine function of past counts and the stability conditions require a different interpretation \citep{fokianos2011log}. The affine form also allows a credibility-esque interpretation, which is lost under the log-link choice. Nevertheless, given that log-linear models are common in actuarial practice, a model that incorporates covariates in a log-linear manner while preserving the affine form of the INGARCH model is desirable.

The second issue relates to the modeling of stationarity. Stationarity serves as a useful and interpretable relaxation of the random-effects assumption. However, in the presence of time-varying covariates, the usual notion of stationarity for the observed process is no longer natural, since the distribution of the observations is allowed to vary systematically over time through these covariates. In such a setting, stationarity is more appropriately formulated at the level of an underlying latent state process after separating the time-varying scale of the observations from the temporal dynamics. This is particularly relevant in insurance applications, where policyholder characteristics, contract features, and exposure-related quantities often change over time.

Many of these issues arise from the fact that the conditional mean in the INGARCH model is trying to satisfy two roles at once: both exhibiting an \emph{evolution} over time, as well as determining the \emph{scale} of the observations. One way to avoid such issues is to disentangle the \emph{level} of the conditional mean  from the \emph{evolution} of the underlying state, specifically by augmenting the model with a priori rates $\lambda_t$. 
By allowing
\[
\E{Z_t\mid Z_{\le t-1}}=\lambda_t M_{t\mid t-1},
\]
we separate the overall scale of the observations, captured by the a priori rates \(\lambda_t\), from the temporal dynamics, captured by the evolving process \((M_{t\mid t-1})_{t\ge1}\). 
In this way, \(M_{t\mid t-1}\) no longer needs to coincide directly with the conditional mean of \(Z_t\). Rather, it can be interpreted as the predictive mean of an underlying latent risk state that governs the temporal evolution. To this end, we formally introduce state-space models (SSMs) \citep{durbin2012time}, in which a process of the latent states is specified alongside the process of observations. This is the topic of the next section.

\section{Practical considerations for insurance ratemaking using the NB-INGARCH model: heterogeneity and variance stationarity}\label{sec.4}

This section addresses the issues discussed in Section~\ref{sec.2.2} that arise when applying the NB-INGARCH model in Model~\ref{mod.1} to insurance ratemaking. 
We begin with the first issue, namely the incorporation of time-varying covariates. To this end, we introduce an extension of the NB-INGARCH model that accommodates time-varying heterogeneity and that is naturally formulated on the positive time index set \(t\in\mathbb{N}\), as is appropriate for ratemaking in property and casualty insurance. 
We refer to this extension as the \emph{heterogeneous NB-INGARCH model}. In particular, this formulation also allows for the degenerate case in which the negative-binomial mean is zero, so that the corresponding observation is almost surely equal to zero.

\begin{model}[Heterogeneous NB-INGARCH model]\label{mod.3}
Given the exogenous sequences
\begin{equation}\label{eq.q1}
(\beta_t^{[0]})_{t\in\mathbb{Z}}\subseteq [0,\infty), 
\qquad 
(\beta_t^{[1]})_{t\in\mathbb{Z}}\subseteq (0,\infty),
\qquad 
(\beta_t^{[2]})_{t\in\mathbb{Z}}\subseteq (0,\infty),
\qquad 
 (\lambda_t)_{t\in\mathbb{Z}}\subseteq [0,\infty),
\end{equation}
we define the heterogeneous NB-INGARCH model for $(Z_t)_{t\in\mathbb{Z}}$ through the following specifications:
\begin{itemize}
  \item[i.] \textbf{Evolution equation:} The sequence of the baseline conditional mean $(M_{t+1\mid t})_{t\in\mathbb{Z}}$ satisfies the recursion in \eqref{eq.nb1}.

  \item[ii.] \textbf{Observation equation:}
  \(
  Z_{t+1}\mid Z_{\le t}\sim \operatorname{NB}\bigl(\kappa_{t+1\mid t},\pi_{t+1\mid t}\bigr)
  \)
  with both $\kappa_{t+1\mid t}>0$ and $\pi_{t+1|t}\in(0,1)$ 
  being known functions of \(Z_{\le t}\) as well as of the exogenous sequences (with this dependence suppressed in the notation), such that
  \[
  \E{Z_{t+1}\mid Z_{\le t}}=\lambda_{t+1}M_{t+1\mid t}.
  \]
  \item [iii.] \textbf{Anchoring at time $t=1$:} Assume
  \begin{align*}
      Z_1 \sim \operatorname{NB}(\kappa_{1|0}, \pi_{1|0}),
  \end{align*}
  with $\kappa_{1|0}>0$ and $\pi_{1|0}\in(0,1)$ being functions of exogenous sequences satisfying
  \[
  \E{Z_1}=\lambda_1 M_{1|0},
  \]
  with $M_{1|0}:=1$.\footnote{The baseline conditional mean at time $t=1$, $M_{1|0}$ is redundant given the presence of the a priori rate $\lambda_1$, leading to identifiability issues if both are estimated from the data, hence the need for the additional restriction $M_{1|0}=1$.}
\end{itemize}
\end{model}

As in Model \ref{mod.1}, 
the heterogeneous NB-INGARCH model above is not uniquely determined by the mean recursion in \eqref{eq.nb1}. 
While this formulation successfully incorporates time-varying heterogeneity, fully specifying the NB-INGARCH model requires an additional restriction on $\kappa_{t+1\mid t}$ and $\pi_{t+1\mid t}$. 

Among various choices of such specifications, this section aims to explore those which accommodate weak stationarity even under the presence of time-varying heterogeneity. Indeed, as briefly mentioned in Section \ref{sec.2.2}, this heterogeneity is incompatible with the weak stationarity of the observations $(Z_t)_{t\ge 1}$. Nevertheless, by identifying the SSM equivalent to the heterogeneous NB-INGARCH model, one can instead discuss stationarity in terms of the latent states, whose scale is invariant to the time-varying heterogeneity. This is our plan in the remainder of this section, with the next subsection beginning with a review of the notion of an M-SSM.

\subsection{Review of the marginalized state-space model}

Compared to the SSM later introduced as Model \ref{mod.100} in Section \ref{sec.5}, the time evolution in the M-SSM\footnote{This framework is referred to as an \emph{observation-driven state-space model} in the statistical literature \citep{brockwell2002introduction}. We use the term \emph{marginalized state-space model} to emphasize that only the conditional marginal distributions of the latent states are specified, without an explicit description of their joint evolution.} takes the simple structure defined by only the conditional marginal distribution of the latent state at each time point given past observations \citep{brockwell2002introduction}. 
In this sense, an M-SSM can be viewed
as a partially specified SSM in which the full transition law is replaced by
its conditional marginal counterpart.

In the following, we present the M-SSM in its most general form. 
Concrete examples of such models in actuarial applications can be found in the recent literature \citep{ahn2025observation, ahn2025observation2}.
We generically use $f$ for the (conditional) distribution (and/or the density) of the considered random variables.

\begin{model}[M-SSM]
\label{mod.2}
An M-SSM for
$(Z_t)_{t\in\mathbb{Z}}$ and $(\Theta_t)_{t\in\mathbb{Z}}$ is defined by:

\begin{itemize}
  \item[i.] An observation equation: 
  \[
  f(Z_t\mid \Theta_t),\qquad t\in\mathbb{Z}, 
  \] 
  together with the Markov assumption
  \begin{equation}\label{eq.e1}
    f(Z_t\mid Z_{\le t-1},\Theta_{\le t})
    =
    f(Z_t\mid \Theta_t),
    \qquad t\in\mathbb{Z};
  \end{equation}

  \item[ii.] A marginal evolution equation:
  \begin{equation}\label{eq.e13}
    f(\Theta_{t+1}\mid Z_{\le t}),
    \qquad t\in\mathbb{Z}.
  \end{equation}
\end{itemize}
\end{model}

 Similar to the INGARCH model, we may equip Model \ref{mod.2} with the additional linear evolution property 
\begin{equation}\label{eq.e231}
\E{\Theta_{t+1}\mid  Z_{\le t}}=\Delta_t\,\E{\Theta_{t}\mid  Z_{\le t}}+c_t,
\end{equation}
for some constants $\Delta_t\in\Real$ and $ c_t\in\Real$. When so equipped, we may refer to  the model as a \textbf{marginalized linear state-space model} (M-LSSM).

Given the sparse specification of Model~\ref{mod.2}---which relies only on the observation equation and the marginal evolution equation, without specifying the conditional joint distribution of the latent states---one may first ask whether the corresponding filtering recursion and predictive densities are well defined. Lemma \ref{cor.e2} shows that they are indeed well defined. However, as discussed in Remark~\ref{rem.1}, this limited specification comes at the cost of not being able to identify the covariance structure of the latent process.

\begin{lemma}\label{cor.e2}
Under Model \ref{mod.2}, we have the following recursive identities:
\begin{itemize}
  \item[i.] The filtering update is given by
  \begin{equation}\label{eqq.2}
  f(\Theta_{t+1}\mid Z_{\le t+1})
  \;\propto\;
  f(Z_{t+1}\mid \Theta_{t+1})\,f(\Theta_{t+1}\mid Z_{\le t}).
  \end{equation}

  \item[ii.] Conditional on  $Z_{\le 0}$, the density of
  $Z_{1:t}$ can be written as
  \[
  \begin{aligned}
  f(Z_{1:t}\mid Z_{\le 0})
  &= \prod_{j=1}^{t} f(Z_j\mid Z_{\le j-1}) \\
  &= \prod_{j=1}^{t}\int
  f(Z_j\mid \theta_j)\,f(\theta_j \mid Z_{\le j-1})\,
  {\rm d}\nu(\theta_j),
  \end{aligned}
  \]
  where $\nu$ is the $\sigma$-finite measure with respect to which the
  conditional density $f(\theta_j\mid Z_{\le j-1})$ is defined.
\end{itemize}
\end{lemma}

\begin{remark}\label{rem.1}
Unfortunately, the M-SSM specification does not allow a full examination of the weak stationarity of the latent process, because it leaves the joint distribution of successive latent states---and hence their covariance structure---unspecified. 
More precisely, the M-SSM determines only the conditional marginal laws
\[
f(\Theta_{t+1}\mid Z_{\le t})
\qquad\text{and}\qquad
f(\Theta_t\mid Z_{\le t}),
\qquad t\in\mathbb{Z}.
\]
The former is specified by the marginal evolution equation in \eqref{eq.e13},
whereas the latter is obtained from the filtering update in Lemma~\ref{cor.e2}.
However, the M-SSM does not specify how these two conditional marginals are
coupled; equivalently, the conditional joint law
\[
f(\Theta_t,\Theta_{t+1}\mid Z_{\le t})
\]
remains unspecified.
%Indeed, the M-SSM identifies only the two conditional marginal laws
%\[
%f(\Theta_{t+1}\mid Z_{\le t})
%\qquad\text{and}\qquad
%f(\Theta_t\mid Z_{\le t}),
%\qquad t\in\mathbb{Z},
%\]
%where the former is specified by the marginal evolution equation in \eqref{eq.e13}, and the latter is obtained from the filtering update in Lemma \ref{cor.e2}, leaving the conditional joint distribution
%\[
%f(\Theta_t,\Theta_{t+1}\mid Z_{\le t})
%\]
%unspecified.  
Consequently, the covariance structure of $(\Theta_t)_{t\in\mathbb{Z}}$ is not identified under the M-SSM formulation. Concrete examples illustrating the non-identifiability of the covariance structure are provided in Section~\ref{sec.3.2}.
\end{remark}

\subsection{State-space formulation of the heterogeneous NB-INGARCH model}

In this subsection, we first introduce a Poisson-Gamma M-LSSM that incorporates
time-varying heterogeneity. We then show
in Proposition~\ref{prop.2} that, under a given restriction on the
heterogeneous NB-INGARCH model, it is equivalent to this Poisson-Gamma M-LSSM.

\begin{example}\label{ex.2}
Let
\begin{equation}\label{exo.1}
(\Delta_t)_{t\in\mathbb{Z}}\subseteq (0,\infty) 
 \quad(c_t)_{t\in\mathbb{Z}}\in[0, \infty), \quad\hbox{and}\quad
 (\lambda_t)_{t\in\mathbb{Z}}\subseteq [0,\infty)
\end{equation}
be exogenous sequences. We define the Poisson-Gamma M-LSSM for a restricted time interval $t\in\mathbb{N}$ with the following specifications:
  \begin{itemize}
    \item[i.] \textbf{Observation equation:}
    \[
      Z_t \mid \Theta_t \;\sim\; \operatorname{Pois}(\lambda_t\Theta_t),
    \]
    \item[ii.]\textbf{Evolution equation:}
    \(
      \Theta_{t+1} \mid Z_{1:t} \;\sim\; \operatorname{Gamma}\bigl(a_{t+1\mid t},\,b_{t+1\mid t}\bigr),
    \)
    where
\[
a_{t+1\mid t}\equiv a_{t+1\mid t}(Z_{\le t})>0
\]
is a prespecified positive function of the past observations \(Z_{\le t}\) as well as the exogenous sequences, and
\[
b_{t+1\mid t}>0
\]
is a prespecified positive function of the exogenous sequences. These quantities are assumed to be chosen so that the linear evolution relation \eqref{eq.e231} holds for the given sequences in \eqref{exo.1}.
 \item [iii.] \textbf{Anchoring at time $t=1$:} Assume
  \begin{align*}
      \Theta_1 \sim \operatorname{Gamma}(a_{1|0}, b_{1|0}),
  \end{align*}
  with $a_{1|0}>0$ and $b_{1|0}=a_{1|0}$.\footnote{Similarly to the restriction $M_{1|0}=1$ in Model \ref{mod.3}, identifiability issues may arise without this restriction.} 
\end{itemize}
The filtering distribution is obtained as
  \[
  a_{t}=a_{t\mid t-1} + Z_{t}\quad \hbox{ and }\quad
  b_{t}=b_{t\mid t-1} + \lambda_t.
  \]

\end{example}

The following result establishes the equivalence between the Poisson-Gamma M-LSSM in Example \ref{ex.2} and the heterogeneous NB-INGARCH model in Model \ref{mod.3} under the additional restriction \eqref{eq.a3.gen}.

\begin{proposition}\label{prop.2}

Fix exogenous sequences
\[
(\beta_t^{[0]},\beta_t^{[1]},\beta_t^{[2]},\lambda_t)_{t\in\mathbb N}, 
\]
as in \eqref{eq.q1}.
Let $(Z_t)_{t\in\mathbb N}$ satisfy the heterogeneous NB-INGARCH specification as in Model \ref{mod.3} with the following further condition
\begin{equation}\label{eq.a3.gen}
\pi_{t\mid t-1}
=
\frac{\lambda_t\beta_t^{[1]}}
{\beta_t^{[2]}+\lambda_t\beta_t^{[1]}},
\qquad t\in\mathbb N.
\end{equation} 
Then there exists a Poisson-Gamma M-LSSM of the form in Example~\ref{ex.2}
which induces this heterogeneous NB-INGARCH model. If further we have
$\lambda_t>0$ for all $t\in\mathbb{N}$, then this Poisson-Gamma M-LSSM is unique.
\end{proposition}

\begin{proof}
We start with the proof of existence.
Note that \eqref{eq.a3.gen} together with the condition
\(
\pi_{t\mid t-1}\in(0,1)
\)
implies that
\(
\lambda_t>0
\)
for all \(t\in\mathbb N\).
Let $(Z_t)_{t\in\mathbb N}$ satisfy the heterogeneous NB-INGARCH specification as in
Model~\ref{mod.3} together with condition \eqref{eq.a3.gen}. Consider the Poisson-Gamma M-LSSM in Example \ref{ex.2} with
\[
a_{t\mid t-1}:=\kappa_{t\mid t-1}>0,
\qquad
b_{t\mid t-1}:=\frac{\beta_t^{[2]}}{\beta_t^{[1]}}>0,
\qquad t\in\mathbb N,
\]
and
\[
\Delta_t:=\beta_t^{[2]}+\lambda_t\beta_t^{[1]}>0,
\qquad
c_t:=\beta_t^{[0]}\ge 0,
\qquad t\in\mathbb N.
\]
Then \(a_{t\mid t-1}\) is a function of \(Z_{1:(t-1)}\), whereas \(b_{t\mid t-1}\),
\(\Delta_t\), and \(c_t\) depend only on the exogenous coefficient sequences.
Hence these quantities have the structural form required in Example~\ref{ex.2}.
In particular, \(b_{t\mid t-1}\) is not a function of the observations.

Moreover, for \(t=1\), since
\[
\E{Z_1}=\lambda_1 M_{1\mid 0}=\lambda_1
\qquad \text{ and } \qquad
Z_1\sim \operatorname{NB}(\kappa_{1\mid 0},\pi_{1\mid 0}),
\]
we have
\[
\lambda_1=\frac{\kappa_{1\mid 0}\pi_{1\mid 0}}{1-\pi_{1\mid 0}}.
\]
Using \eqref{eq.a3.gen} at \(t=1\) gives us
\[
\lambda_1
=
\frac{\kappa_{1\mid 0}\lambda_1\beta_1^{[1]}}{\beta_1^{[2]}},
\]
hence
\[
a_{1\mid 0}
=
\kappa_{1\mid 0}
=
\frac{\beta_1^{[2]}}{\beta_1^{[1]}}
=
b_{1\mid 0}.
\]
Therefore, the anchoring condition in Example~\ref{ex.2} is satisfied.

Moreover, by \eqref{eq.a3.gen},
\[
\frac{\lambda_t}{b_{t\mid t-1}+\lambda_t}
=
\frac{\lambda_t}{\beta_t^{[2]}/\beta_t^{[1]}+\lambda_t}
=
\frac{\lambda_t\beta_t^{[1]}}
{\beta_t^{[2]}+\lambda_t\beta_t^{[1]}}
=
\pi_{t\mid t-1},
\qquad t\in\mathbb N.
\]
Since under Example~\ref{ex.2},
\[
Z_t\mid Z_{1:(t-1)}
\sim
\operatorname{NB}\!\left(a_{t\mid t-1},\,\frac{\lambda_t}{b_{t\mid t-1}+\lambda_t}\right),
\qquad t\in\mathbb N,
\]
where \(t=1\) is understood through the anchoring distribution, we obtain that the predictive distribution of \(Z_t\mid Z_{1:(t-1)}\) under the Poisson-Gamma M-LSSM coincides with that of the given heterogeneous NB-INGARCH model.

Further,
\[
\lambda_t\frac{a_{t\mid t-1}}{b_{t\mid t-1}}
=
\lambda_t\frac{\kappa_{t\mid t-1}}{b_{t\mid t-1}}
=
\frac{\kappa_{t\mid t-1}\,\pi_{t\mid t-1}}{1-\pi_{t\mid t-1}}
=
\E{Z_t\mid Z_{1:(t-1)}}
=
\lambda_t M_{t\mid t-1},
\qquad t\in\mathbb N.
\]
Hence
\[
\frac{a_{t\mid t-1}}{b_{t\mid t-1}}=M_{t\mid t-1},
\qquad t\in\mathbb N.
\]

It remains to verify the linear evolution equation \eqref{eq.e231}. Since
\[
\E{\Theta_t\mid Z_{1:t}}
=
\frac{a_{t\mid t-1}+Z_t}{b_{t\mid t-1}+\lambda_t},
\qquad
\frac{a_{t\mid t-1}}{b_{t\mid t-1}}
=
\frac{\kappa_{t\mid t-1}\,\pi_{t\mid t-1}}{\lambda_t(1-\pi_{t\mid t-1})}
=
M_{t\mid t-1},
\]
we have
\[
\begin{aligned}
\Delta_t \E{\Theta_t\mid Z_{1:t}}+c_t
&=
\Delta_t \frac{a_{t\mid t-1}+Z_t}{b_{t\mid t-1}+\lambda_t}+c_t\\
&=
c_t + \frac{\Delta_t}{b_{t\mid t-1}+\lambda_t}\,Z_t
+ \Delta_t\frac{b_{t\mid t-1}}{b_{t\mid t-1}+\lambda_t}\,M_{t\mid t-1}\\
&=
\beta_t^{[0]}+\beta_t^{[1]}Z_t+\beta_t^{[2]}M_{t\mid t-1}\\
&=
M_{t+1\mid t}\\
&=
\E{\Theta_{t+1}\mid Z_{1:t}}.
\end{aligned}
\]
Hence the linear evolution relation \eqref{eq.e231} holds. Therefore, the Poisson-Gamma M-LSSM in
Example~\ref{ex.2} with these choices induces the stated heterogeneous NB-INGARCH model.
This proves existence.

\medskip

Now, we prove uniqueness.
Assume there exists a Poisson-Gamma M-LSSM of the form as in Example~\ref{ex.2} inducing the given
heterogeneous NB-INGARCH model. Since under Example~\ref{ex.2},
\[
      Z_t \mid \Theta_t \;\sim\; \operatorname{Pois}(\lambda_t\Theta_t)
      \quad\hbox{and}\quad \Theta_t\mid Z_{1:(t-1)}\sim
\operatorname{Gamma}\bigl(a_{t\mid t-1},\,b_{t\mid t-1}\bigr),
\]
the predictive distribution of \(Z_t\mid Z_{1:(t-1)}\) is
\[
Z_t\mid Z_{1:(t-1)}
\sim
\operatorname{NB}\!\left(a_{t\mid t-1},\,\frac{\lambda_t}{b_{t\mid t-1}+\lambda_t}\right),
\qquad t\in\mathbb N,
\]
where \(t=1\) is again understood through the anchoring distribution.
Comparing with Model~\ref{mod.3}, we obtain
\[
a_{t\mid t-1}=\kappa_{t\mid t-1}>0,
\qquad
\frac{\lambda_t}{b_{t\mid t-1}+\lambda_t}=\pi_{t\mid t-1},
\qquad t\in\mathbb N.
\]
Since \(\lambda_t>0\), this gives
\[
b_{t\mid t-1}
=
\lambda_t\frac{1-\pi_{t\mid t-1}}{\pi_{t\mid t-1}}
=
\frac{\beta_t^{[2]}}{\beta_t^{[1]}}>0,
\qquad t\in\mathbb N,
\]
where the last equality follows from \eqref{eq.a3.gen}. In particular,
\(b_{t\mid t-1}\) is uniquely determined by the exogenous coefficient sequences and is free of the observations.
Also, for \(t=1\), the anchoring condition in Example~\ref{ex.2} yields
\[
a_{1\mid 0}=b_{1\mid 0}=\kappa_{1\mid 0}.
\]

Finally, comparing the affine recursion
\[
M_{t+1\mid t}
=
c_t + \frac{\Delta_t}{b_{t\mid t-1}+\lambda_t}\,Z_t
+ \Delta_t\,\frac{b_{t\mid t-1}}{b_{t\mid t-1}+\lambda_t}\,M_{t\mid t-1}
\]
with \eqref{eq.nb1}, we obtain
\[
\Delta_t=\beta_t^{[2]}+\lambda_t\beta_t^{[1]}>0
\qquad \text{ and } \qquad
c_t=\beta_t^{[0]}\ge 0.
\]
Hence, \(a_{t\mid t-1}\), \(b_{t\mid t-1}\), \(\Delta_t\), and \(c_t\) are uniquely
determined. Therefore, the induced Poisson-Gamma M-LSSM is unique.
This completes the proof.
\end{proof}

We now have the equivalence between the heterogeneous NB-INGARCH model and Poisson-Gamma M-LSSM. Hence, the discussion of the weak stationarity under the framework of the heterogeneous NB-INGARCH model can be equivalently addressed through the framework of the Poisson-Gamma M-LSSM model.

\subsection{Mean and variance stationarity of the latent states}

While, in Example \ref{ex.2},  we specify the evolution of the baseline conditional mean
\[
M_{t+1\mid t}=\frac{a_{t+1\mid t}}{b_{t+1\mid t}}
\]
through the linear relation \eqref{eq.e231}, the individual dynamics of \(a_{t+1\mid t}\) and \(b_{t+1\mid t}\) have not yet been prescribed. Example \ref{ex.3} below provides one such specification, where \(a_{t+1\mid t}\) and \(b_{t+1\mid t}\) are chosen according to \eqref{constant variance_q}. 

Among the various possible specifications for the dynamics of \(a_{t+1\mid t}\) and \(b_{t+1\mid t}\), including those considered in \citet{zhu2011negative} and \citet{xu2012model}, the specification below is particularly appealing for insurance ratemaking because it yields latent states \((\Theta_t)_{t\in\mathbb{N}}\) whose mean and variance remain constant over time; see Proposition~\ref{prop.4q} below, taken from \citet{ahn2025observation}. Specifically, the linear evolution condition \eqref{eq.e231} with \(c_t=1-\Delta_t\) guarantees mean stationarity of the Poisson-Gamma M-LSSM in Example~\ref{ex.3}, namely
\[
\E{\Theta_t}\equiv \E{\Theta_1}=1,
\qquad t\ge 1.
\]
It therefore remains to ensure that the variance also remains constant, and the choice of \(q_t\) in Example~\ref{ex.3} is made precisely for this purpose. While Proposition~\ref{prop.4q} guarantees mean and variance stationarity of the latent states, the covariance structure of the latent process remains unspecified under the model definition. A full discussion of weak stationarity, including the covariance structure, is deferred to Section \ref{sec.3.5} and Section~\ref{sec.5}, where alternative ways of specifying the joint evolution are discussed.

\begin{example}\label{ex.3}
Under the same setting of Example \ref{ex.2} with the additional restriction $\Delta_t\le 1$ for all $t\in\mathbb{N}$, replace the evolution equation (item ii) with the following:
  \begin{itemize}
    \item[ii.]\textbf{Evolution equation:}
    \(
      \Theta_{t+1} \mid Z_{1:t} \;\sim\; \operatorname{Gamma}\bigl(a_{t+1\mid t},\,b_{t+1\mid t}\bigr),
    \)
    where
    \[
    b_{t+1\mid t} = q_{t} \left(b_{t|t-1} + \lambda_{t}\right) \,\,\text{ and }\,\,\,
    a_{t+1\mid t} = \Delta_{t} q_{t} \left( a_{t|t-1} +Z_{t}\right)
    +(1-\Delta_{t})b_{t+1\mid t},
    \]
    and where
\begin{equation}\label{constant variance_q}
q_{t} \;=\; \dfrac{1}{\,\Delta_{t}^2 + (1-\Delta_{t}^2)\,\dfrac{b_{t|t-1}+\lambda_{t}}{a_{1|0}}\,}.
\end{equation}
  \end{itemize}
These specifications guarantee that the linear evolution relation \eqref{eq.e231} holds for the given sequences in \eqref{exo.1}, with $c_t = 1-\Delta_t$.
\end{example}

\begin{proposition}\citep{ahn2025observation}
\label{prop.4q}
Under the setting of Example~\ref{ex.3}, we have for all  $t\ge 1$:
\begin{itemize}
  \item[i.] The mean of the latent state is given by  $\E{\Theta_{t}}=1$.
  \item[ii.] The variance of the latent state is given by
  \(
        \Var{\Theta_{t}} \;=\; 1/a_{1|0}.
  \) 
  \item[iii.] The mean of the observation is given by
  $\E{Z_{t}}=\lambda_{t}$.
  \item[iv.] The variance of the observation is given by
  \(
        \Var{Z_{t}}
        \;=\; \lambda_{t} \;+\; \frac{(\lambda_{t})^{\,2}}{a_{1|0}}.
  \)
\end{itemize}
    
\end{proposition}

We now introduce the heterogeneous NB-INGARCH model corresponding to the Poisson-Gamma M-LSSM in Example~\ref{ex.3}, in which the parameters $\kappa_{t+1\mid t}$ and $\pi_{t+1\mid t}$ are fully specified. The equivalence between the two formulations is established in Proposition~\ref{prop.4} below. In contrast to Proposition \ref{prop.2}, no additional restriction is required here, since the relevant constraint has already been built into the model specification. This is the main model proposed in this paper for insurance ratemaking. In the following subsection, we discuss the weak-stationarity property of this model.

\begin{model}[heterogeneous NB-INGARCH model for insurance ratemaking]\label{mod.4}
Given the exogenous specification
\begin{equation}\label{exo.2}
(\Delta_t)_{t\in\mathbb{N}}\subseteq (0,1],
\qquad
(\lambda_t)_{t\in\mathbb{N}}\subseteq [0,\infty),\qquad 
a_{1|0}>0,
\end{equation}
 define recursively the positive sequence
\((b_{t\mid t-1})_{t\in\mathbb{N}}\) by
\[
b_{t+1\mid t}
=
\begin{cases}
a_{1\mid 0}, & t=0;\\
q_t\left(b_{t\mid t-1}+\lambda_t\right), & t\in\mathbb{N},
\end{cases}
\]
where $q_t$ is defined in \eqref{constant variance_q}.
Also define
\[
\kappa_{t+1\mid t}
:=
\begin{cases}
a_{1\mid 0}, & t=0;\\
\Delta_t q_t\left(\kappa_{t\mid t-1}+Z_t\right)
+
(1-\Delta_t)b_{t+1\mid t}, & t\in\mathbb{N},
\end{cases}
\]
and
\[
\pi_{t\mid t-1}:=\frac{\lambda_t}{b_{t\mid t-1}+\lambda_t},
\qquad t\in\mathbb{N}.
\]

We define the heterogeneous NB-INGARCH model for insurance ratemaking for
\((Z_t)_{t\in\mathbb{N}}\) by the following specifications:
\begin{itemize}
  \item[i.] \textbf{Evolution equation:}
  The sequence of the baseline conditional mean
  \((M_{t\mid t-1})_{t\in\mathbb{N}}\) is defined by
  \[
  M_{t+1\mid t}
  :=\begin{cases}
  1, & t=0;\\
  \beta_t^{[0]}+\beta_t^{[1]}Z_t+\beta_t^{[2]}M_{t\mid t-1}, & t\in\mathbb{N},
  \end{cases}
  \]
where
  \[
  \beta_t^{[0]}:=1-\Delta_t,
  \qquad
  \beta_t^{[1]}:=\frac{\Delta_t}{b_{t\mid t-1}+\lambda_t},
  \qquad
  \beta_t^{[2]}:=\Delta_t\frac{b_{t\mid t-1}}{b_{t\mid t-1}+\lambda_t}.
  \]

  \item[ii.] \textbf{Observation equation:}
  Conditionally on \(Z_{1:(t-1)}\), 
  \[
  Z_t\mid Z_{1:(t-1)}
  \sim
  \operatorname{NB}\bigl(\kappa_{t\mid t-1},\pi_{t\mid t-1}\bigr).
  \]
  Hence, we have
  \(
  \E{Z_t\mid Z_{1:(t-1)}}=\lambda_t M_{t\mid t-1}
  \).

  \item[iii.] \textbf{Anchoring at time \(t=1\):} Initialize
  \(
  Z_1\sim \operatorname{NB}(\kappa_{1\mid 0},\pi_{1\mid 0})
  \).
  Hence, we have
  \(
  M_{1\mid 0}=1
  \).
\end{itemize}
\end{model}

We omit the proof of Proposition \ref{prop.4} as it is analogous to that of Proposition \ref{prop.2}.

\begin{proposition}\label{prop.4}
Let $(Z_t)_{t\in\mathbb N}$ satisfy the heterogeneous NB-INGARCH specification as in Model \ref{mod.4}.
Then there exists a Poisson-Gamma M-LSSM of the form in Example~\ref{ex.3}
which induces this heterogeneous NB-INGARCH model. If further we have
$\lambda_t>0$ for all $t\in\mathbb{N}$, then this Poisson-Gamma M-LSSM is unique.
\end{proposition}

\subsection{Discussion on the weak stationarity of the heterogeneous NB-INGARCH model}
\label{sec.3.5} 

While the Poisson-Gamma M-LSSM in Example~\ref{ex.3} guarantees that the latent states have time-invariant means and variances, as shown in Proposition~\ref{prop.4q}, it cannot itself determine its covariance structure.
The difficulty is that an M-LSSM specifies only the conditional marginal distributions of the latent states as shown in Remark \ref{rem.1}, the covariance, \(\cov{\Theta_t,\Theta_{t+k}}\), is not identified within the M-LSSM formulation alone.

As will be shown in Section~\ref{sec.5}, any M-SSM admits an SSM\footnote{The formal definition of the SSM would be introduced in Model~\ref{mod.100} in Section \ref{sec.5}.} extension by additionally specifying a conditional joint evolution law for successive latent states that is consistent with the marginal evolution prescribed by the M-SSM.
The key question is whether these extensions are consistent in the sense that, although their full joint dynamics may differ, they all induce the same covariance structure for the latent process. If so, then the covariance structure---and hence weak stationarity---is well defined for the Poisson-Gamma M-LSSM in Example~\ref{ex.3} despite the non-uniqueness of the extension. 
Although the model is already suitable for likelihood-based inference and prediction, this additional result is needed to fully justify its use particularly in short-term insurance applications, by establishing the weak stationarity of the latent states.
The following theorem, which is the main result of this paper, addresses this question.
The proof of this theorem appears at the end of Section \ref{sec.5}.

\begin{theorem}
\label{thm.1}
For the process
$(Z_t, \Theta_t)_{t\in \mathbb{N}}$ following the heterogeneous M-LSSM of Example \ref{ex.3}, there exists a state-space model\footnote{See Model \ref{mod.100} in Section \ref{sec.5} for the definition of the state-space model.} representation $(Z_{t}^\ast,\Theta_{t}^\ast)_{t\in\mathbb{N}}$  satisfying, for \(Z_{1:t}\)-a.e.~observation path \(z_{1:t}\),
      \[
        \Theta_{t}  \;\big|\, Z_{1:t}=z_{1:t}
        \eqd
        \Theta_{t}^\ast \;\big|\, Z_{1:t}^\ast=z_{1:t}, 
        \quad\quad
        \Theta_{t+1} \;\big|\, Z_{1:t}=z_{1:t}
        \eqd
        \Theta_{t+1}^\ast \;\big|\, Z_{1:t}^\ast=z_{1:t},
      \]
      and
      \begin{align}
\E{\Theta_{t+1}^*|\Theta_t^*, Z_{1:t}=z_{1:t}}=\Delta_t \Theta_t^* + (1-\Delta_t).
\label{eq.c1}
      \end{align}
      Furthermore, the covariance structure of $(\Theta_{t}^*)_{t\ge 1}$ is uniquely determined by $\Theta_t$ as
\begin{align*}
\cov{\Theta^*_t, \Theta^*_{t+k}}&= \Var{\Theta_t}\prod_{j=0}^{k-1}\Delta_{t+j}.
\end{align*}
 Hence, the weak stationarity of $\Theta_{ t}$ holds provided that $\Delta_{ t} = \Delta_{ 1}$ for all $t \ge 1$. 
\end{theorem}

For the interpretation of Theorem~\ref{thm.1}, we use the following terminology. 
For any given M-SSM, an SSM representation
\((Z_t^\ast,\Theta_t^\ast)_{t\in\mathbb N}\) is called a \emph{lift} if it preserves the conditional marginal laws of the latent states specified by the M-SSM; that is, for \(Z_{1:t}\)-a.e.~observation path \(z_{1:t}\),
\[
\Theta_t\mid Z_{1:t}=z_{1:t}
\eqd
\Theta_t^\ast\mid Z_{1:t}^\ast=z_{1:t},
\qquad
\Theta_{t+1}\mid Z_{1:t}=z_{1:t}
\eqd
\Theta_{t+1}^\ast\mid Z_{1:t}^\ast=z_{1:t}.
\]

Using this terminology, Theorem~\ref{thm.1} shows that the Poisson-Gamma M-LSSM in Example~\ref{ex.3}, whose latent states already have time-invariant mean and variance, admits a lift satisfying the condition \eqref{eq.c1} in which the covariance structure of the latent process is well defined. Moreover, while the full joint evolution in the lift need not be unique, this covariance structure is unique among all such lifts. Hence, weak stationarity is not an artifact of a particular choice of lift. In particular, if
\(\Delta_t\equiv \Delta\) for all \(t\ge 1\), then
\[
\cov{\Theta_t^\ast,\Theta_{t+k}^\ast}
=
\frac{\Delta^k}{a_{1\mid 0}},
\]
so the lifted latent state process is weakly stationary. Consequently, based on the model equivalence in Proposition~\ref{prop.4}, the heterogeneous NB-INGARCH model in Model~\ref{mod.4} admits an unambiguous weakly stationary latent-state interpretation under this specification.

The next section develops the state-space representation needed to justify
Theorem~\ref{thm.1}.  Readers mainly interested in practical implementation may skip
directly to Section~\ref{sec.6}.

\section{State-space representations and the weak stationarity of latent states}
\label{sec.5}

This section introduces the SSM, in which the joint evolution of successive latent states is explicitly specified, as well as its linear version: the \emph{linear state-space model} (LSSM). See \eqref{eq.e23} below for the formal definition of LSSM.
As shown in Section~\ref{sec.3.2} below, any M-(L)SSM admits a lift by specifying a conditional copula for successive latent states.
However, when the lift is restricted to LSSM, i.e., equipped with a certain linear evolution property, it is no longer automatic that there exists a lift satisfying the desired linear evolution property. The first goal of this section is therefore to characterize when a given M-LSSM admits such a linear lift, namely an LSSM that induces the same conditional marginal laws as the M-LSSM and satisfies the corresponding linear conditional-mean evolution.

Such a linear lift provides a covariance structure for the latent process and therefore makes it possible to discuss weak stationarity of the latent states. However, this covariance structure is not intrinsic to the M-LSSM itself, since the M-LSSM specifies only the conditional marginal laws and not the joint distribution of successive latent states. The covariance structure is induced by any particular LSSM lift is just one among many possible such lifts. Consequently, for weak-stationarity statements to be well defined under the M-LSSM formulation, one must show that all admissible LSSM lifts induce the same covariance structure. This is the second goal of this section. If this holds, then the covariance structure can be regarded as uniquely determined by the original M-LSSM, even though the full joint evolution need not be unique.

We first develop the general lifting result for M-LSSMs. We then specialize it to the Poisson-Gamma M-LSSM in Example~\ref{ex.2} and \ref{ex.3}, or equivalently the NB-INGARCH models in Models~\ref{mod.3} and~\ref{mod.4}.

\subsection{The linear state-space model}

The specification of an M-SSM leaves out a coherent model for the joint distribution of the latent state-observation process, $(\Theta_t, Z_t)_{t \in \mathbb{Z}}$. This shortcoming is addressed in the following definition, which fully specifies these dynamics.

\begin{model}[SSM]
\label{mod.100}
\noindent
The stochastic process consisting of observable responses $(Z_t)_{t\in\mathbb{Z}}$ and latent states $(\Theta_t)_{t\in\mathbb{Z}}$ is an SSM if, for all $t\in\mathbb{Z}$, it satisfies
\begin{equation}\label{eq.e1B}
f(Z_t \mid Z_{\le t-1}, \Theta_{\le t}) =f(  Z_t \mid \Theta_{t}),
\end{equation}
and
\begin{equation}\label{eq.e21}
f(\Theta_{t+1} \mid Z_{\le t}, \Theta_{\le t})=f(\Theta_{t+1} \mid Z_{\le t}, \Theta_{t}).
\end{equation}
\end{model}

\bigskip

Analogously to item ii of Lemma \ref{cor.e2}, we have the following result showing the well-definedness of the SSM.

\begin{corollary}\label{cor.e1}
Under the setting of Model \ref{mod.100}, 
the joint densities, for $t\ge 1$, are given by
\[
f(Z_{\le t},\Theta_{\le t}) \;=\;
%\begin{cases}
\displaystyle
f(Z_{\le 0} , \Theta_{\le 0})\,
\prod_{j=1}^{t}
  f(Z_j\mid\Theta_j)\,
\prod_{j=0}^{t-1}  f(\Theta_{j+1}\mid Z_{\le j},\Theta_{j}).
\]
The corresponding one-step filtering update is given  by
\[
f(\Theta_{t+1}\mid Z_{\le t+1}, \Theta_{t})
\;\propto\;
%\begin{cases}
f(Z_{t+1}\mid\Theta_{t+1})\,f(\Theta_{t+1}\mid Z_{\le t},\Theta_{t}).
%,   & \text{for an SSM},\\[4pt]
%f(Z_{t+1}\mid\Theta_{t+1})\,\pi(\Theta_{t+1}\mid \Theta_{t}),           & \text{for a P-SSM}.\\
%\end{cases}
\]
\end{corollary}

\bigskip

As explained in Corollary \ref{cor.e1}, an SSM is completely specified by the two laws
\begin{equation}\label{observation density}
f(Z_t \mid \Theta_{t})  \quad \text{and} \quad
f(\Theta_{t+1} \mid Z_{\le t}, \Theta_{t}).
\end{equation}
Clearly, these correspond respectively to the observation equation \eqref{eq.e1} and the marginal evolution equation \eqref{eq.e21} of the M-SSM. The former retains its role and name as the {\it observation equation}. As the latter now explicitly specifies the manner in which one arrives at the current state from the preceding ones, we call it the {\it full evolution equation}  to emphasize its distinction from its counterpart in \eqref{eq.e13} of the M-SSM. 

With the joint distribution of the latent states fully specified, Model \ref{mod.100} may be equipped with the following additional linear evolution property, which is a refinement of \eqref{eq.e231}:
\begin{equation}\label{eq.e23}
\E{\Theta_{t+1}|Z_{\leq t}, \Theta_t} = \Delta_t \Theta_t + c_t.
\end{equation}
When so equipped, the model is referred to as an LSSM. Our main interest is in the second-order stationarity of the LSSMs.
The following result is useful for the calculation of the covariance of the latent states as well as the observable responses in LSSMs.

\begin{lemma}\label{cov.structure}
Consider an LSSM satisfying \eqref{eq.e23} with evolution coefficients $(\Delta_t)_{t \in \mathbb{Z}}$ and $(c_t)_{t\in\mathbb{Z}}$, and with finite second moments $\E{\Theta_t^2}<\infty$ for all $t$.
\begin{itemize}
%  \item[i.] $\cov{\Theta_t, \Theta_{t+1}} \;=\; \Delta_t\,\Var{\Theta_t}$.
  \item[i.] For any $k\ge 1$,   we have
\begin{equation}\label{eq.cov0}
\cov{\Theta_t, \Theta_{t+k}} =\Var{\Theta_t}\,\prod\limits_{j=0}^{k-1}\Delta_{t+j}.
\end{equation}
  \item[ii.] If additionally the observation equation satisfies $\E{Z_u\mid \Theta_u}=\lambda_u\Theta_u$ for all $u\in\mathbb{Z}$ and for some exogenous constants $(\lambda_u)_{u\in\mathbb{Z}}\subseteq [0, \infty)$, then for any $k\ge 1$,
  \[
    \cov{Z_t, Z_{t+k}} \;=\; \lambda_t \lambda_{t+k}\cov{\Theta_t, \Theta_{t+k}}.
  \]
\end{itemize}

\end{lemma}
\begin{proof}
Under the above assumptions, for $k\ge 1$ we have the recursion 
\[
\begin{aligned}
\E{\Theta_{t+k}\mid \Theta_t, Z_{\le t}}%&=\E{\Theta_{t+k}\mid \Theta_{\le t}, Z_{\le t}}\\
&=\E{\E{\Theta_{t+k}\mid \Theta_{\le t+k-1}, Z_{\le t+k-1}}\mid \Theta_{t}, Z_{\le t}}\\
&=\Delta_{t+k-1}\, \E{\Theta_{t+k-1}\mid \Theta_{ t}, Z_{\le t}} +c_{t+k-1}.
\end{aligned}
\]
Recursively expanding, we have
\begin{equation}\label{eq.condexp}
\E{\Theta_{t+k}\mid \Theta_{t}, Z_{\le t}}=\Theta_t \prod\limits_{j=0}^{k-1}\Delta_{t+j} +c_{t:t+k-1}^*,
\end{equation}
where we set
\[
c_{t:t+k-1}^*
\;=\;
\sum_{r=0}^{k-1}
      c_{t+r}
      \prod_{j=r+1}^{k-1} \Delta_{t+j},
\]
with an empty product set to one.
As a result, applying the tower property, we have
\[
\E{\Theta_{t+k}\mid Z_{\le t}}=\E{\Theta_t\mid Z_{\le t}} \,\prod\limits_{j=0}^{k-1}\Delta_{t+j} +c_{t:t+k-1}^*,
\]
and the serial covariance structure
\[
\cov{\Theta_t, \Theta_{t+k}} = \cov{\Theta_t, \E{\Theta_{t+k}\mid \Theta_t, Z_{\le t}}}
= \cov{\Theta_t, \Theta_t \,\prod\limits_{j=0}^{k-1}\Delta_{t+j}}
=\Var{\Theta_t}\,\prod\limits_{j=0}^{k-1}\Delta_{t+j}.
\]

\noindent
For item~ii, using \eqref{eq.condexp}, we have
\[
\begin{aligned}
\cov{Z_t, Z_{t+k}}
&= \cov{Z_t, \E{ Z_{t+k}  \mid  Z_{\le t+k-1}, \Theta_{\le t+k}} }=\cov{Z_t, \lambda_{t+k}\Theta_{t+k}}\\
&= \lambda_{t+k}\cov{Z_t, \E{\Theta_{t+k}|Z_{\leq t}, \Theta_{\leq t}}}\\
&= \lambda_{t+k}\cov{Z_t, \Theta_t \prod\limits_{j=0}^{k-1}\Delta_{t+j} +c_{t:t+k-1}^*}
= \lambda_{t+k}\left(\prod\limits_{j=0}^{k-1}\Delta_{t+j}\right)\cov{Z_t, \Theta_t}.
\end{aligned}
\]
The claim then follows from item i, because
\[
\cov{Z_t, \Theta_t} = \cov{\E{Z_t|\Theta_t}, \Theta_t} = \cov{\lambda_t \Theta_t, \Theta_t}=\lambda_t\Var{\Theta_t}.
\]
This completes the proof.
\end{proof}

\subsection{Identifiability of the SSM corresponding to an M-SSM} 
\label{sec.3.2}
Every SSM induces the corresponding M-SSM by retaining the same observation equation and marginalizing the full evolution equation over the current latent state. More precisely, if an SSM has full evolution density
\[
f(\Theta_{t+1}\mid Z_{\le t},\Theta_t),
\]
then the marginal evolution equation of the induced M-SSM is obtained as
\begin{equation}\label{eq.52}
f(\Theta_{t+1}\mid Z_{\le t})
=
\int
f(\Theta_{t+1}\mid Z_{\le t},\Theta_t=\theta_t)\,
f(\theta_t\mid Z_{\le t})\,
{\rm d}\nu(\theta_t),
\end{equation}
where \(f(\Theta_t\mid Z_{\le t})\) is the filtering density of $\Theta_t$ induced by the SSM and \(\nu\) is the reference measure for the latent state. Thus, the marginal evolution equation in \eqref{eq.e13} is precisely the conditional marginal law obtained by integrating out \(\Theta_t\) from the full state evolution.

Conversely, a given M-SSM can always be lifted to at least one SSM.
For this, assume an M-SSM with the following conditional marginal laws
\begin{equation}\label{eq.51}
f(\Theta_t\mid Z_{\le t})
\quad\hbox{and}\quad
f(\Theta_{t+1}\mid Z_{\le t}),
\end{equation}
where the former is obtained from the filtering recursion and the latter is specified by the marginal evolution equation. With these two marginal laws, any admissible conditional copula, given \(Z_{\le t}\), defines a compatible conditional joint law:
\begin{equation}\label{eq.53}
f(\Theta_t,\Theta_{t+1}\mid Z_{\le t})=c\!\left(
F(\Theta_t\mid z_{\le t}),
F(\Theta_{t+1}\mid Z_{\le t})
\mid Z_{\le t}
\right) f(\Theta_{t}\mid Z_{\le t}) 
f(\Theta_{t+1}\mid Z_{\le t}),
\end{equation}
where \(F(\cdot\mid Z_{\le t})\) denotes the corresponding conditional cumulative distribution function, and \(c(\cdot,\cdot\mid Z_{\le t})\) is the density of the chosen (absolutely continuous) conditional copula.
As a result, this gives the full evolution equation
\[
\begin{aligned}
f(\Theta_{t+1}\mid Z_{\le t},\Theta_t)
&=
\frac{
f(\Theta_t,\Theta_{t+1}\mid Z_{\le t})
}{
f(\Theta_t\mid Z_{\le t})
}  \\
&=
c\!\left(
F(\Theta_t\mid Z_{\le t}),
F(\Theta_{t+1}\mid Z_{\le t})
\,\middle|\, Z_{\le t}
\right)
f(\Theta_{t+1}\mid Z_{\le t}).
\end{aligned}
\]
Since the conditional copula in \eqref{eq.53} preserves the specified conditional marginals, marginalizing \eqref{eq.53} over \(\Theta_t\) recovers the original marginal evolution density \(f(\Theta_{t+1}\mid Z_{\le t})\).

The following example gives one such construction for a particular Poisson-Gamma M-SSM of the form considered in Example~\ref{ex.2}. This construction was originally proposed by \citet{harvey1989time} and has subsequently been used in actuarial applications \citep{bolance2007greatest, abdallah2016sarmanov, youn2023simple}. Following the original formulation, the model is initialized by specifying the distribution of the latent state at \(t=1\).

\begin{example} \citep{harvey1989time}\label{ex.w13_}
We define the SSM with the following specifications:
  \begin{description}

    \item[i. \textbf{Observation equation:}]
    \(
      Z_{t} \mid \Theta_{t} \;\sim\; \operatorname{Pois}(\Theta_{t}).
    \)

    \item[ii. \textbf{Evolution equation:}]
For given 
    \(
      \Theta_{t} \mid Z_{1:t} \;\sim\; \operatorname{Gamma}\bigl(a_{t},\,b_{t}\bigr),
    \)
    where $a_t:=a_t(Z_{1:t})$ and $b_t:=b_t(Z_{1:t})$ are positive functions of the observable responses $Z_{1:t}$,
    define the evolution
    \[
    \Theta_{t+1}:=\frac{\Theta_tB_{t+1}}{q},
    \]
    where $q\in(0,1)$ is a constant and
    \[
    B_{t+1}\mid B_{1:t}, Z_{1:t}, \Theta_{1:t} \sim \operatorname{Beta}(qa_t, (1-q)a_t).
    \]
Thus, $\Theta_{t+1}$ is a \emph{Beta-thinned} and scaled version of $\Theta_t$. This yields
\begin{equation}\label{eq.o1}
      \Theta_{t+1} \mid Z_{1:t} \;\sim\; \operatorname{Gamma}\bigl(a_{t+1\mid t},\,b_{t+1\mid t}\bigr),
      \qquad\text{where}\qquad
      a_{t+1\mid t}:=qa_t,\quad b_{t+1\mid t}:=qb_t.
\end{equation}
    \item[iii. \textbf{Anchoring at $t=1$}:] Assume
  \begin{equation}\label{eq.o0}
  \Theta_{1} \;\sim\;
              \operatorname{Gamma}\bigl(a_{1\mid 0},\,b_{1\mid 0}\bigr),
  \end{equation}
  with constants $a_{1\mid 0}= b_{1\mid 0}>0$, so that $\mu_{1\mid 0}:=\E{\Theta_1}=1$.
  \end{description}

Using the Poisson-Gamma conjugate property and the evolution of the parameters in \eqref{eq.o1} along with the anchoring assumption
in \eqref{eq.o0}, we have the following recursion
\[
a_{t+1\mid t}=q\bigl(a_{t\mid t-1}+Z_t\bigr),
\quad\hbox{and}\quad 
b_{t+1\mid t}=
q\bigl(b_{t\mid t-1}+1\bigr), \qquad t\in\mathbb{N}_0,
\]
and hence note that $b_{t+1|t}$ is free of the observations while $a_{t+1|t}$ is a function of $Z_{1:t}$.

Now, we derive the moments of $(\Theta_t)_{t\ge 1}$ and $(Z_t)_{t\ge 1}$ using the evolution of the parameters in  \eqref{eq.o1} along with the anchoring assumption
in \eqref{eq.o0}. For example, by the tower property we have
\begin{equation}\label{eq.o2}
\E{\Theta_t}=1
\qquad\text{and}\qquad
\E{Z_t}=1,
\qquad t\ge 1.
\end{equation}
Further, we have
\[
\begin{aligned}
\Var{\Theta_{t+1}}
&=\Var{\E{\Theta_{t+1}\mid Z_{1:t}}} + \E{\Var{\Theta_{t+1}\mid Z_{1:t}}}\\
&=\Var{\frac{a_t}{b_t}} + \frac{1}{q}\E{\frac{a_t}{b_t^2}}\\
&=\Var{\Theta_t} + \left(\frac{1}{q}-1 \right)\frac{1}{b_t}\E{\frac{a_t}{b_t}} \\
&=\Var{\Theta_t}
+
\frac{1-q}{q\,b_t},
\end{aligned}
\]
where the second equality used the fact that $b_t$ is a function of $q$ and $t$, and the last equality is from the tower property applied to $\Var{\Theta_t}$.

Importantly, since $\E{B_{t+1}\mid Z_{1:t}}=q$, the evolution satisfies
\begin{equation}\label{eq.evo}
\E{\Theta_{t+1}\mid \Theta_t, Z_{1:t}}
=
\Theta_t,
\end{equation}
hence this SSM exhibits a linear evolution (i.e., \eqref{eq.e23}) with $\Delta_t\equiv 1$ and $c_t\equiv 0$.
Thus, Lemma \ref{cov.structure} gives, for $t,k\ge 1$,
\begin{equation}\label{eq.o5}
\cov{\Theta_t,\Theta_{t+k}}=\Var{\Theta_t}
\qquad\text{and}\qquad
\cov{Z_t,Z_{t+k}}=\cov{\Theta_t,\Theta_{t+k}}=\Var{\Theta_t}.
\end{equation}
This closes the example.
\end{example}

Integrating out $\Theta_t$ from the distribution of $\Theta_{t+1}|Z_{1:t}, \Theta_t$ yields Example \ref{ex.3} with the specifications $\lambda_t \equiv 1$ and $\Delta_t\equiv 1$, but with the alternate specification $q_t \equiv  q$ for some $q\in(0,1)$. This alternate specification of $q_t$ means that variance stationarity is no longer guaranteed, and indeed it increases over time. 

Importantly, an M-SSM does not determine an SSM lift uniquely. 
Indeed, different admissible choices of the conditional copula, or equivalently of the conditional joint distribution in \eqref{eq.53}, may lead to different SSMs inducing the same M-SSM, as the following example along with Example~\ref{ex.w13_} shows. In particular, the covariance structure of \((\Theta_t)_{t\ge 1}\) may depend on the chosen lift.
Consequently, the mere existence of an SSM lift is not sufficient to make weak-stationarity statements well defined under an M-SSM specification. 

\begin{example}\label{ex.4}
The evolution equation in Example \ref{ex.w13_} can be replaced, e.g., by 
\begin{description}
      \item[ii. \textbf{Evolution equation:}]
For given $q\in(0,1)$ and 
    \(
      \Theta_{t} \mid Z_{1:t} \;\sim\; \operatorname{Gamma}\bigl(a_{t},\,b_{t}\bigr),
    \)
    where $a_t:=a_t(Z_{1:t})$ and $b_t:=b_t(Z_{1:t})$ are positive functions of the observable responses $Z_{1:t}$,
    define the evolution
    \[
    \Theta_{t+1}:=\Theta_{t+1}^*,
    \]
    where
    \[
    \Theta_{t+1}^*\mid B_{1:t}, Z_{1:t}, \Theta_{1:t} \sim \operatorname{Gamma}(qa_t,\,qb_t).
    \]
\end{description}
A straightforward calculation shows that Example~\ref{ex.w13_} with this alternative evolution equation
induces the same \emph{marginal} evolution as in Example \ref{ex.3}.

However, the \emph{dependence} structure is different. In particular, by construction we have in the latter evolution equation assumption
\[
\Theta_{t+1} \perp \Theta_t \ \big|\ Z_{1:t},
\]
which gives for $t\ge 1$
\[
\begin{aligned}
\cov{\Theta_t,\Theta_{t+1}}&=\cov{\Theta_t,\E{\Theta_{t+1}\mid\Theta_t, Z_{1:t}}}\\
&=\cov{\Theta_t, \E{\Theta_t\mid Z_{1:t}}}\\
&=\cov{\E{\Theta_t\mid Z_{1:t}}, \E{\Theta_t\mid Z_{1:t}}}\\
&=\Var{\E{\Theta_t\mid Z_{1:t}}}\\
&<\Var{\E{\Theta_t\mid Z_{1:t}}} + \frac{1}{b_t}=\Var{\Theta_t},
\end{aligned}
\]
where the first three equalities are from the tower property of covariance, and the last equality is from the tower property of variance and mean stationarity.
\end{example}

The above examples also illustrate how the full linear evolution property \eqref{eq.e23} is \emph{stricter} than the marginal linear evolution property \eqref{eq.e231}. In particular, although their marginal evolution equations coincide, it is the fact that the two versions of Example \ref{ex.w13_} and Example \ref{ex.4} differ in their full linear evolution properties which allows the covariance structure specified to differ. The full evolution equation of Example \ref{ex.w13_} is given by \eqref{eq.evo}, i.e., it satisfies the linear evolution property \eqref{eq.e23} with $\Delta_t=1$ and $c_t=0$, which leads to the covariance of the form in \eqref{eq.o5} induced by Lemma \ref{cov.structure}. However, Example \ref{ex.4} does not have the same covariance as in \eqref{eq.o5}, which, with the help of Lemma \ref{cov.structure}, implies that they do not match in their full linear evolution coefficients.

\subsection{Lifting an M-LSSM to an LSSM}
\label{Lifting  M-LSSM to LSSM}
Our next goal is to lift an M-LSSM satisfying the marginal linear relation \eqref{eq.e231} to an LSSM that satisfies the linear conditional-mean property  \eqref{eq.e23}, both with identical sequences $(\Delta_t, c_t)_{t\in\mathbb{Z}}$. Note that \eqref{eq.e23} implies \eqref{eq.e231}.
As discussed in Section \ref{sec.3.2}, one can lift any M-SSM to an SSM by specifying the copula. However, if one additionally imposes the condition \eqref{eq.e23} on the SSM, then 
 the existence of such a lift is no longer guaranteed. 
 
 This subsection seeks conditions under which such a lift exists. Lemma \ref{cor:varNec} in \ref{app} shows that a necessary requirement is the conditional convex ordering
\[ \left[\Delta_t\Theta_t+c_t \mid Z_{\le t}\right]
\;\preceq_{\mathrm{cx}}\;
\left[\Theta_{t+1}\mid Z_{\le t}\right], \qquad \text{a.s.,}
\]
for given constants $\Delta_t, c_t\in\Real$.
If this convex ordering fails to hold on a set of histories $z_{\le t}$, with positive probability, then no SSM lift with the conditional mean as in \eqref{eq.e23} exists. In particular, since the convex ordering implies the variance ordering, provided second moments exist, we also have the following conditional variance inequality
\begin{align*}
\Var{\Delta_t\Theta_t+c_t \mid Z_{\le t}}
  \;\le\;
  \Var{\Theta_{t+1} \mid Z_{\le t}}, \qquad \text{a.s.}
\end{align*}

Interestingly, this conditional convex ordering is also sufficient for the existence of the lift as Theorem \ref{thm.3} below shows. 
This follows from the classical result of \citet{strassen1965coupling}, or more precisely \citet{leskela2017condconvex}. 
The proof of Theorem \ref{thm.3} as well as the full statement of the result in \citet{leskela2017condconvex} can be found in \ref{app.2}.

\begin{theorem}\label{thm.3}
Let \((Z_t)_{t\in \mathbb{Z}}\) and \((\Theta_t)_{t\in \mathbb{Z}}\) form an M-LSSM (i.e., Model~\ref{mod.2} satisfying \eqref{eq.e231}, $Z_{\leq t}$-a.s.).
%(i.e., Model~\ref{mod.2}).
%Assume for every $t\in\mathbb{Z}$, that the conditional means $\E{\Theta_{t+1} \mid Z_{\le t}}$ are finite, a.s., and let $(m_t(\cdot))_{t\in \mathbb{Z}}$ be a series of $Z_{\leq t}$-measurable functions which satisfy \eqref{eq.e231}, $Z_{\leq t}$-a.s.
Then, the following are equivalent:

\begin{enumerate}
\item[i.] There exist processes \((Z_t^\ast)_{t\in\mathbb{Z}}\) and \((\Theta_t^\ast)_{t\in\mathbb{Z}}\) forming an LSSM (i.e., Model \ref{mod.100} satisfying \eqref{eq.e23}) such that, for \(Z_{\le t}\)-a.e.~observation path \(z_{\le t}\),
      \[
        \Theta_t  \;\big|\, Z_{\le t}=z_{\le t}
        \;\eqd\;
        \Theta_t^\ast \;\big|\, Z_{\le t}^\ast=z_{\le t}
        \qquad \text{ and }\qquad
        \Theta_{t+1} \;\big|\, Z_{\le t}=z_{\le t}
        \;\eqd\;
        \Theta_{t+1}^\ast \;\big|\, Z_{\le t}^\ast=z_{\le t},
      \]
      and with the evolution equation
      \[
        \E{\Theta_{t+1}^\ast\mid \Theta_t^\ast, Z_{\le t}}=\Delta_t\Theta_t^\ast+c_t
        %m_t(\Theta_t^\ast)
        ,\quad\text{a.s.}
      \]

\item[ii.] The M-LSSM satisfies
      \[
        \left[\Delta_t\Theta_t+c_t\mid Z_{\le t}\right]\;\preceq_{\mathrm{cx}}\;\left[\Theta_{t+1}\mid Z_{\le t}\right].
      \]
\end{enumerate}
\end{theorem}

Theorem~\ref{thm.3} does not guarantee the uniqueness of 
the lift from an M-LSSM to an LSSM. 
However, the next Corollary~\ref{thm.cov.uniq} shows that the covariance structure induced by such a lift is unique.
Consequently, in applications where the latent dependence is essential 
(e.g., heterogeneity), weak-stationarity statements for 
$(\Theta_t)_{t\in\mathbb Z}$ become well posed. 
In particular, one may analyze weak stationarity using any LSSM lift, 
without committing to a specific choice of conditional copula.
We finally note that this contrasts with the more general lift from M-SSM to SSM discussed in 
Section~\ref{sec.3.2}, where different lifts can lead to different covariances among the latent states.

\begin{corollary}[Uniqueness of the covariance structure under an LSSM lift]
\label{thm.cov.uniq}
Consider the two processes $(Z_t）_{t\in\mathbb{Z}}$ and $(\Theta_t)_{t\in\mathbb{Z}}$ under the setting of the M-LSSM of Model~\ref{mod.2} satisfying \eqref{eq.e231} with sequences $(\Delta_t)_{t\in \mathbb{Z}}$ and $(c_t)_{t\in\mathbb{Z}}$, and assume
\[
\Var{\Theta_t}<\infty,\qquad t\in\mathbb{Z}.
\]
Assume that condition {\it ii}.~of Theorem~\ref{thm.3} holds with the same choice of $(\Delta_t)_{t\in \mathbb{Z}}$ and $(c_t)_{t\in\mathbb{Z}}$, so that existence of the corresponding LSSM lift is guaranteed by Theorem \ref{thm.3}. Then, the covariance structure of the latent states $(\Theta_t)_{t\in\mathbb{Z}}$
does not depend on the particular choice of the LSSM lift, i.e., for any LSSM process
  $(Z_t^\ast,\Theta_t^\ast)_{t\in\mathbb{Z}}$ satisfying condition~{\it i}., the covariance structure is uniquely determined as
  \begin{equation}\label{eq.cov.uniq.thm}
  \begin{aligned} \cov{\Theta_t^\ast,\Theta_{t+k}^\ast}\;=\;\Biggl(\prod_{j=0}^{k-1}\Delta_{t+j}\Biggr)\,\Var{\Theta_t},
  \end{aligned}
  \end{equation}
  for all $t\in\mathbb{Z}$ and $k\ge 1$.

\end{corollary}

\begin{proof}

By condition {\it i}.~of Theorem~\ref{thm.3}, the tower property shows that the processes $(Z_s^\ast)_{s\le t}$ and $(Z_s)_{s\le t}$ satisfy
\[
\E{h(\Theta_t)}
=
\E{\E{h(\Theta_t)\mid Z_{\le t}}}
=
\E{\E{h(\Theta_t^\ast)\mid Z_{\le t}^\ast}}
=
\E{h(\Theta_t^\ast)},
\]
for any $h$ such that $\E{h(\Theta_t)}<\infty$. Specifically, since this holds for $h(x)=x^k$ with $k=1, 2$, we have
\begin{equation}\label{eq.var.match}
\Var{\Theta_t^\ast}=\Var{\Theta_t}.
\end{equation}
Next, since $(Z_t^\ast,\Theta_t^\ast)_{t\in\mathbb{Z}}$ is an LSSM lift with conditional mean
\[
\E{\Theta_{t+1}^\ast\mid \Theta_t^\ast, Z_{\le t}^\ast}=\Delta_t\Theta_t^\ast+c_t,
\]
Lemma~\ref{cov.structure} thus yields, for any $k\ge 1$,
\[
\cov{\Theta_t^\ast,\Theta_{t+k}^\ast}
=
\Var{\Theta_t^\ast}\,\prod_{j=0}^{k-1}\Delta_{t+j}.
\]
Combining this identity with \eqref{eq.var.match} proves \eqref{eq.cov.uniq.thm}.
\end{proof}

\subsection{Existence of an LSSM representation for the NB-INGARCH}

With Theorem \ref{thm.3} established, we are finally ready to lift the M-LSSM of Example \ref{ex.2}, equivalently the heterogeneous NB-INGARCH model in Model \ref{mod.3}, to its analogous LSSM version. Note that the proof of the necessary part of the following Theorem \ref{thm.gam} provides one of such explicit constructions of the desired LSSMs.

\begin{theorem}
\label{thm.gam}
Under the setting of the M-LSSM in Example~\ref{ex.2} with exogenous sequences in \eqref{exo.1}, we have the
conditional convex-order relation
\begin{equation}\label{eq.a1}
   \left[\Delta_t \Theta_t + c_t\bigm|Z_{\le t}\right]
      \;\preceq_{\mathrm{cx}}\;
   \left[\Theta_{t+1}\bigm|Z_{\le t}\right],
\end{equation}
\emph{if and only if}
\begin{equation}\label{eq:cond_general}
   0~\le~ \Delta_t\frac{b_{t+1\mid t}}{b_t}~\le~ 1.
\end{equation}
\end{theorem}

\begin{proof}
%The proof for $\Delta_t=0$ is trivial. Therefore, we focus on the case $\Delta_t>0$.
First, we have
\[
   \mathbb{E}[\Delta_t \Theta_t + c_t\mid Z_{\le t}]
      =\Delta_t\,\frac{a_t}{b_t}+c_t,
   \qquad
   \mathbb{E}[\Theta_{t+1}\mid Z_{\le t}]
      =\frac{a_{t+1\mid t}}{b_{t+1\mid t}},
\]
and
\begin{equation}\label{eq.a2}
\frac{a_{t+1\mid t}}{b_{t+1\mid t}}
      =\Delta_t\,\frac{a_t}{b_t}+c_t.
\end{equation}
Set
\[
   X=\Delta_t\Theta_t+c_t,
   \qquad
   Y=\Theta_{t+1},
   \qquad
   q_t=\frac{b_{t+1\mid t}}{b_t}>0.
\]
We first prove the sufficient part. With the mean restriction \eqref{eq.a2}, the conditional moment generating functions are
\[
   \E{e^{sX}\mid Z_{\le t}}=e^{c_t s}\Bigl(1-\frac{\Delta_t s}{b_t}\Bigr)^{-a_t}
   \qquad \text{ and }\qquad
   \E{e^{sY}\mid Z_{\le t}}=\Bigl(1-\frac{s}{q_t b_t}\Bigr)^{-q_t(\Delta_t a_t+c_t b_t)}.
\]
Let $D(s; Z_{\le t})=\log \E{e^{sY}\mid Z_{\le t}}-\log \E{e^{sX}\mid Z_{\le t}}$. Then,
\[
\begin{aligned}
  D(s; Z_{\le t})=a_t \log \left(1-\frac{\Delta_t s}{b_t}\right) -  c_t s - q_t (\Delta_t a_t + c_t b_t) \log \left(1-\frac{s}{q_t b_t}\right).
     \end{aligned}
\]
Suppose, contrary to the assumption, that $q_t>1/\Delta_t$. Then, as $s \to (b_t/\Delta_t)^+$,
\begin{align*}
    D(s; Z_{\leq t})&\to a_t \log 0^- -c_t \frac{b_t}{\Delta_t}-q_t (\Delta_t a_t + c_t b_t)\log \left(1-\frac{1}{\Delta_t q_t}\right)=-\infty.
\end{align*}
Hence, for $s_{\epsilon} = b_t/\Delta_t - \epsilon$, with $\epsilon>0$ sufficiently small, we must have $D(s_{\epsilon}; Z_{\leq t})<0$. Since we can find such an $\epsilon>0$ for any realization of $Z_{\leq t}$, the sufficiency is proved.

Now, we prove the necessary part. First, note that
\begin{align*}
\Delta_t \Theta_t|Z_{\leq t}\sim {\rm Gamma}(a_t, b_t/\Delta_t).
\end{align*}
Let $X^{\prime}=\Delta_t\Theta_t + V_t$, where $V_t|\Theta_{\leq t}, Z_{\leq t}\sim {\rm Gamma}(c_t b_t/\Delta_t, b_t/\Delta_t)$. By construction, we have $\Delta_t \Theta_t \perp V_t|Z_{\leq t}$, hence $X^{\prime}\sim {\rm Gamma}(a_t + c_t b_t /\Delta_t, b_t/\Delta_t)$. Next, note that by the mean restriction \eqref{eq.a2}, and by the assumption $\Delta_t b_{t+1|t}/b_t \leq 1$,
\begin{align*}
a_{t+1|t}&= \frac{\Delta_t b_{t+1|t}}{b_t}(a_t + c_t b_t/\Delta_t)\leq a_t + c_t b_t/\Delta_t.
\end{align*}
Hence, let $Q_t|\Theta_t, V_t, Z_{\leq t}\sim {\rm Beta}(a_{t+1|t}, a_t + c_t b_t/\Delta_t - a_{t+1|t})$ (with a ${\rm Beta}(k, 0)$ with $k>0$ understood to be a degenerate distribution at 1), and set
\begin{align*}
Y^*&= \frac{(b_t/\Delta_t)Q_t X^{\prime}}{b_{t+1|t}}.
\end{align*}
By a Beta-thinning, $Q_t X^{\prime}|Z_{\leq t}\sim {\rm Gamma}(a_{t+1|t}, b_t/\Delta_t)$, such that $Y^*|Z_{\leq t}\sim {\rm Gamma}(a_{t+1|t}, b_{t+1|t})\stackrel{(d)}{=}Y|Z_{\leq t}$. We have
\begin{align*}
\E{Y^*|X, Z_{\leq t}}= \E{Y^*|\Theta_t, Z_{\leq t}}
=\frac{(b_t/\Delta_t)\E{Q_t|Z_{\leq t}}}{b_{t+1|t}}\, (\Delta_t \Theta_t + \E{V_t|Z_{\leq t}})= \Delta_t \Theta_t + c_t = X,
\end{align*}
hence by the result of \citet{leskela2017condconvex} (see Theorem \ref{condcoupling} in \ref{app.2}), we have $\left[X|Z_{\leq t}\right]\preceq_{cx}\left[Y^*|Z_{\leq t}\right]$, which also implies $\left[X|Z_{\leq t}\right]\preceq_{cx}\left[Y|Z_{\leq t}\right]$. This completes the proof of the necessary part.
\end{proof}

\medskip

Finally, applying Theorem~\ref{thm.gam} to the specification in Example~\ref{ex.3} yields the proof of Theorem~\ref{thm.1}.

\begin{proof}[Proof of Theorem \ref{thm.1}]
From Theorem~\ref{thm.gam}, the lift required in Theorem \ref{thm.1} exists whenever $0\le \Delta_{t}q_t\le 1$ for all $t\ge 1$; by Lemma \ref{lem.1} in \ref{app}, this is satisfied whenever $0< \Delta_t \leq 1$ for all $t \geq 1$, and then Corollary \ref{thm.cov.uniq} concludes the proof. 
\end{proof}

\section{Practical application of the heterogeneous NB-INGARCH model}
\label{sec.6}

This section illustrates the likelihood-based estimation of the proposed
heterogeneous NB-INGARCH model in a property and casualty insurance
ratemaking application.

\subsection{Illustration of the model estimation procedure}

We use the heterogeneous NB-INGARCH model in Model~\ref{mod.4}, or
equivalently its Poisson-Gamma M-LSSM representation in Example~\ref{ex.3}.
Following Theorem~\ref{thm.1}, we impose
\[
\Delta_t \equiv \Delta \in (0,1],
\]
so that the model admits an LSSM lift whose latent state process is
weakly stationary.

Model~\ref{mod.4} is formulated for a single policyholder. To apply it to
panel data, we introduce an additional index \(i=1,\ldots,n\) for
policyholders and assume that the corresponding processes are independent
across \(i\). For \(i=1,\ldots,n\) and \(t=1,\ldots,T\), we specify the
a priori rate as
\[
\lambda_{i,t}
=
\exp\left(\mathbf{x}_{i,t}^{\top}\boldsymbol w\right),
\]
where \(\boldsymbol w\in\Real^{p+1}\) is the vector of regression
coefficients and \(\mathbf{x}_{i,t}\in\Real^{p+1}\) is the covariate vector
for policyholder \(i\) at time \(t\), including an intercept term.

Under this panel specification, the exogenous quantities in \eqref{exo.2}
are parameterized by
\[
(\boldsymbol w,\Delta,a_{1\mid 0})
\in
\Real^{p+1}\times(0,1]\times\Real_{>0}.
\]
These parameters are estimated from the training data
\[
\left\{(z_{i,t},\mathbf{x}_{i,t}): i=1,\ldots,n,\ t=1,\ldots,T\right\},
\]
where \(z_{i,t}\) denotes the realization of \(Z_{i,t}\).

Let \(\mathrm{dNB}(\cdot;\kappa,\pi)\) denote the negative-binomial
probability mass function with parameters \(\kappa\) and \(\pi\). The
log-likelihood is then given by
\begin{equation}
\begin{aligned}
\ell(\boldsymbol w,\Delta,a_{1\mid 0})
&=
\sum_{i=1}^n \sum_{t=1}^{T}
\log f\left(z_{i,t}\mid z_{i,1:(t-1)}\right)  \\
&=
\sum_{i=1}^n \sum_{t=1}^{T}
\log\left\{
\mathrm{dNB}\left(
z_{i,t};
\kappa_{i,t\mid t-1},
\pi_{i,t\mid t-1}
\right)
\right\},
\end{aligned}
\end{equation}
where \(\kappa_{i,t\mid t-1}\) and \(\pi_{i,t\mid t-1}\) are computed
recursively as in Model~\ref{mod.4}. Parameter estimation is carried out by numerically maximizing this
log-likelihood using standard optimization routines in {\tt R} \citep{R2025}, for example
the {\tt optim} function.

In the next subsection, we  
demonstrate a natural, closed-form treatment of missing data under the framework of the Poisson-Gamma M-LSSM in Example \ref{ex.3}, and hence equivalently under the framework of the heterogeneous NB-INGARCH model in Model \ref{mod.4}.

\subsection{Handling missing data}

The framework in Model \ref{mod.4} allows \(\lambda_{i,t}\) to take the value zero. This can potentially lead to non-identifiability of the INGARCH coefficients, see Proposition \ref{prop.2}, and therefore favors the M-LSSM formulation over the INGARCH formulation.
In particular, allowing \(\lambda_{i,t}=0\) provides a principled and computationally convenient way to handle missing data by treating such time points as having zero exposure. Missing observations whose a priori rates are non-zero remain in the marginal likelihood of the observations \(Z_{i,t}\), thereby requiring integration for likelihood evaluation. We can avoid this complication by instead setting \(\lambda_{i, t^*}=0\) for any \(t^*\in\mathbb{Z}\) such that \(Z_{i,t^*}\) is unobserved. Since this makes \(Z_{i,t^*}\) degenerate at \(0\), we obtain
\[
a_{i,t^*}=a_{i,t^*\mid t^*-1}
\qquad\text{and}\qquad
b_{i,t^*}=b_{i,t^*\mid t^*-1},
\]
that is, the filtering step is skipped at any such $t^*$. 
Thus for such $t^*$, instead of modeling the uncertainty of the missing data through the unknown realization of \(Z_{i,t^*}\), this treatment carries that uncertainty forward only through the evolution step, without adding new information through filtering.

In line with such a treatment of missing observations, we deal with 
missing observations through an indicator \(R_{i,t}\), where
\(R_{i,t}=0\) means that \(z_{i,t}\) is observed and \(R_{i,t}=1\) means that
it is missing. For a missing observation, we set
\[
\lambda_{i,t}=0
\qquad\text{and}\qquad
z_{i,t}=0.
\]
Equivalently, under such a specification in case of missing data \(R_{i,t}=1\), the negative-binomial distribution is then
degenerate at zero, the corresponding likelihood contribution is equal to one, and
\[
\log\left\{
\mathrm{dNB}\left(
z_{i,t};
\kappa_{i,t\mid t-1},
\pi_{i,t\mid t-1}
\right)
\right\}
=0.
\]
Thus, missing observations do not contribute directly to the log-likelihood,
while their effect is propagated through the recursive evolution of the model.
This is a clear advantage of the state-space formulation of the NB-INGARCH model over its original definition.

\subsection{Model performance comparison using real property and casualty insurance data}

We measure the performance of the proposed heterogeneous
NB-INGARCH model using a real insurance dataset. The data is from the
Wisconsin Local Government Property Insurance Fund (LGPIF), described in
\citet{frees2016multivariate}. The dataset contains observed inland marine
(IM) claims from 1,234 local government entities over the years 2006--2010,
which we index by \(t=1,\ldots,5\). Summary statistics for this dataset can
be found in \citet{frees2016multivariate} and \citet{youn2023simple}.

Entities with zero coverage are excluded, leaving 1,016 entities for the
empirical analysis. We include two categorical covariates. The first is
\emph{entity type}, with six levels: miscellaneous, city, county, school,
town, and village. The second is \emph{coverage level}, which is divided into
three intervals:
\[
(0,0.135], \qquad (0.135,0.495], \qquad (0.495,\infty).
\]
In addition, data from \(n_{\rm test}=894\) local government entities in 2011
(indexed by \(t=6\)) are used as a hold-out sample for testing.

For comparison, we consider two benchmark models:
\begin{itemize}
  \item The stationary Poisson-AR(1) state-space model \citep{pinquet2001allowance, ahn2021ordering}, whose latent state process is strictly stationary.
  Here, strict stationarity means that all finite-dimensional distributions of
  the latent state process are invariant under time shifts, which is stronger
  than weak stationarity, where only the first two moments and the covariance
  structure are required to be time invariant.
  
  \item The Poisson-Gamma random-effects model
  \citep{dionne1989generalization,boucher2008models}. This model is nested in
  the heterogeneous NB-INGARCH model in Model~\ref{mod.4} and is obtained by
  imposing the restriction \(\Delta_t\equiv 1\).
\end{itemize}

Table~\ref{tab.1} reports the out-of-sample test MSE and the out-of-sample
test log-likelihood for period \(t=6\). The latter is defined as
\[
\sum_{i=1}^{n_{\rm test}}
\log f\left(z_{i,6}\mid z_{i,1:5}\right).
\]
The results are reported for the proposed heterogeneous NB-INGARCH model and
the benchmark models.

%\begin{table}[t]
%\centering
%\begin{tabular}{lcc}
%\toprule
%\textbf{Model} & \textbf{MSE} & \textbf{log-likelihood} \\
%\midrule
%%Poisson-INGARCH
%%  & 0.0835
%%  & -164.34 \\
%NB–INGARCH
%  & 0.0787
%  & -157.42 \\
%Poisson–Gamma random-effects
%  & 0.0785
%  & -157.11 \\
%\bottomrule
%\end{tabular}
%\caption{Out-of-sample test MSE and Out-of-sample test log-likelihood at $t=6$ (year 2011).}
%\label{tab.1}
%
%
%\end{table}

\begin{table}[t]
\centering
\begin{tabular}{lcc}
  \toprule
  & \multicolumn{2}{c}{\bf out-of-sample}\\
\textbf{Model} & \textbf{MSE} & \textbf{log-likelihood} \\
\midrule
heterogeneous NB–INGARCH
  & 0.193
  & -176.94 \\
Poisson-AR(1) state-space model
  & 0.184
  & -177.90 \\
Poisson–Gamma random-effects
  & 0.197
  & -179.02 \\
\bottomrule
\end{tabular}
\caption{Out-of-sample test MSE and Out-of-sample test log-likelihood at $t=6$ (year 2011).}
\label{tab.1}

\end{table}

Table~\ref{tab.1} shows that the three models have broadly comparable
out-of-sample predictive performance, although the best-performing model
depends on the evaluation criterion. The proposed heterogeneous NB-INGARCH
model attains the highest test log-likelihood, whereas the stationary
Poisson-AR(1) state-space model gives the smallest test MSE. 

From a computational perspective, however, the NB-INGARCH model has an
important advantage over the Poisson-AR(1) state-space model. Its conditional
predictive distribution is available in closed form as a negative-binomial
distribution, so likelihood evaluation does not require numerical integration
or filtering over the latent state. This leads to a simpler and faster
estimation procedure. 
The trade-off is that the Poisson-AR(1) state-space model assumes strict
stationarity of the latent state process, whereas the proposed NB-INGARCH model
establishes weak stationarity of the latent states through its equivalent
state-space representation.

The Poisson-Gamma random-effects model, which corresponds to the restriction
\(\Delta_t\equiv 1\), performs slightly worse than the other two models under
both criteria.
For the heterogeneous NB-INGARCH model, the estimated value
\(\Delta_t \equiv \widehat{\Delta}\approx 0.925\) is large, but less than one. Thus, the fitted
model is not very different from the Poisson-Gamma random-effects model, while still allowing
a modest degree of time-order dependence in the claims history. This suggests
that, for this dataset, there is evidence for dynamic updating beyond a pure
random-effects structure.

\section{Conclusion}\label{sec.7}

This paper bridges INGARCH models and state-space models (SSMs) by showing how the former can be represented within the latter framework and by deriving explicit SSM formulations for the NB-INGARCH cases. The state-space representation improves interpretability, allows covariates to be incorporated under a weakly stationary latent state framework, and provides a convenient treatment of missing observations within the INGARCH modeling paradigm.

Future research may extend this framework to broader classes of INGARCH models. A systematic investigation of weak stationarity within these extensions, especially under time-varying covariates and panel data settings, would further strengthen the practical relevance of the proposed framework in insurance applications.

\bigskip

\section*{Data Availability}
The data and code used in this study are available at:\\
\url{ https://github.com/jaeyoun-ahn/INGARCH}

\section*{Acknowledgements}

Jae Youn Ahn was supported by the National Research Foundation of Korea (NRF) through a grant funded by the Korea government (MSIT) (RS-2025-23524530) and the Basic Science Research Program (Priority Research Institute) funded by the Ministry of Education (2021R1A6A1A10039823).
Parts of this manuscript were written while Mario V.~W\"uthrich was hosted at Ewha Womans University, Seoul.

\section*{Conflicts of interest}

The authors declare no conflicts of interest regarding this manuscript.

\bibliographystyle{apalike}
\bibliography{bib_tex}

\appendix

\section{Miscellaneous results}\label{app}

\begin{lemma}
\label{lem.1}
Under the setting of Example~\ref{ex.3},  we have
\begin{equation}\label{eq:uniform-cond}
\Delta_{t}\,\frac{b_{t+1\mid t}}{b_{t}}\ \le\ 1,
\end{equation}
for every $t\ge 1$.
\end{lemma}

\begin{proof}
Since $\Delta_{t}\leq 1$, it suffices to show that $b_{t+1|t}/b_{t}=q_{t}\leq 1$. In turn, it suffices to show that $b_{t}/a_{1|0}\geq 1$ for all $t$, since if this holds,
\begin{align*}
q_{t}^{-1} &= \Delta^2_{t}+(1-\Delta^2_{ t})\frac{b_{t}}{a_{1|0}}\geq \Delta^2_{t}+(1-\Delta^2_{t})=1.
\end{align*}
We prove this via induction as follows.  First, this holds for $t=1$ since $a_{1|0}=b_{1|0}\leq b_{1|0}+\lambda_{1}=b_{1}$. 
Next, given that $b_{t}/a_{1|0}\geq 1$ for some $t$, 
\begin{align*}
\frac{b_{t+1}}{a_{1|0}}&= q_{t} \frac{b_{t}}{a_{1|0}} + \lambda_{t+1}\\
&\geq \frac{b_{t}/a_{1|0}}{\Delta^2_{t}+(1-\Delta^2_{ t})(b_{t}/a_{1|0})}\\
&= \frac{1}{\Delta^2_{t}(b_{t}/a_{1|0})^{-1}+(1-\Delta^2_{t})}\geq 1.
\end{align*}
This completes the proof.
\end{proof}

\section{Proof of Theorem \ref{thm.3}}\label{app.2}

We first provide the result that a necessary requirement is the conditional convex ordering.

\begin{lemma}\label{cor:varNec}
Consider the stochastic process of observed responses $(Z_t)_{t\in\mathbb{Z}}$ and latent states $(\Theta_t)_{t\in\mathbb{Z}}$ in the M-SSM of Model~\ref{mod.2} satisfying the conditional mean condition \eqref{eq.e231} for given sequences $(\Delta_t)_{t\geq 0}$ and $(c_t)_{t\geq 0}$. %, and assume $\E{\Theta_t^2}+\E{\Theta_{t+1}^2}<\infty$.
Assume there exists an SSM (Model~\ref{mod.100}) with processes $(Z_t^*)_{t\in\mathbb{Z}}$ and $(\Theta_t^*)_{t\in\mathbb{Z}}$ such that
%\[
%  \operatorname{Law}\!\big(\Theta_{t+1}^* \mid Z_{\le t}^*\big)
%  \;=\;
%  \operatorname{Law}\!\big(\Theta_{t+1} \mid Z_{\le t}\big),
%  \qquad \text{a.s.},
%\]
\[
 \Theta_{t+1}^* \mid Z_{\le t}^*
  \;\eqd\;
  \Theta_{t+1} \mid Z_{\le t},
  \qquad \text{a.s.},
\]
and
\begin{equation}\label{eqq.1}
  \E{\Theta_{t+1}^* \mid \Theta_t^*, Z_{\le t}^*} \;=\; \Delta_t\Theta_t^*+c_t,
  \qquad \text{a.s.}
\end{equation}
Then,
\[
        \left[\Delta_t\Theta_t^*+c_t \mid Z_{\le t}^*\right]
\;\preceq_{\mathrm{cx}}\;
\left[\Theta_{t+1}^*\mid Z_{\le t}^*\right], \qquad \text{a.s.}
\]
\end{lemma}

\begin{proof}
Let $m_t(x) = \Delta_t x + c_t$. Firstly, by \eqref{eqq.1} and the tower property, we have
\[
\E{m_t(\Theta_t^*) \mid Z_{\le t}^*} = \E{\Theta_{t+1}^*\mid Z_{\le t}^*}, \qquad \text{a.s.}
\]
It hence remains to show that $\E{\phi( m_t(\Theta_t^*))|Z^*_{\le t}}
\leq \E{\phi(\Theta_{t+1}^*)|Z^*_{\le t}}$ for any convex function $\phi$. By Jensen's inequality,
\[
 \phi( m_t(\Theta_t^*)) = \phi\left(\E{\Theta_{t+1}^*\mid\Theta_t^*, Z_{\le t}^*}\right)\leq \E{\phi(\Theta_{t+1}^*)|\Theta_t^*, Z^*_{\leq t}}, \qquad \text{a.s.}
 \]
 From this the claim follows.
%This implies
%\[ \E{\phi(m_t(\Theta_t^*))|Z^*_{\leq t}}\leq \E{\E{\phi(\Theta_{t+1}^*)|\Theta_t^*, Z^*_{\leq t}}|Z^*_{\leq t}}= \E{\phi(\Theta_{t+1}^*)|Z^*_{\leq t}}.\]
\end{proof}

The proof of the other direction hence, the proof of Theorem \ref{thm.3}, is given as follows.

\begin{proof}[Proof of Theorem \ref{thm.3}]
  \emph{i$\,\Rightarrow\,$ii}: This immediately follows from Lemma \ref{cor:varNec}.

  Now we prove \emph{i\,$\Leftarrow$\,ii}: Let $m_t(x) = \Delta_t \Theta_t + c_t$, and let $X_t := m_t(\Theta_t)$. The proof is trivial for $\Delta_t = 0$, hence we focus on $\Delta_t \neq 0$. 
   From Theorem \ref{condcoupling}, we have the coupling $(X_t^*, \Theta_{t+1}^*, Z_{\le t}^*)$ such that
  \begin{equation}\label{ap.1}
  (X_t, Z_{\le t}) \eqd(X_t^*, Z_{\le t}^*)
  \quad\hbox{and}\quad
  (\Theta_{t+1}, Z_{\le t}) \eqd(\Theta_{t+1}^*, Z_{\le t}^*), \qquad \text{a.s.},
  \end{equation}
  satisfying the condition
  \begin{equation}\label{ap.2}
  \E{\Theta_{t+1}^* \mid X_t^*, Z_{\le t}^*}= X_t^*, \qquad \text{a.s.}
  \end{equation}
  Observe that \eqref{ap.1} implies
  \[
  X_t\mid  Z_{\le t}=z_{\le t} \eqd X_t^*\mid Z_{\le t}^*=z_{\le t}
  \quad\hbox{and}\quad
  \Theta_{t+1}\mid Z_{\le t}=z_{\le t}  \eqd \Theta_{t+1}^*\mid Z_{\le t}^*=z_{\le t}.
  \]
It remains to extend the coupling to a suitable $\Theta^*_t$. Towards this end, note that since $\Delta_t\neq 0$, $m_t(\cdot)$ is one-to-one, with $\Theta_t = m_t^{-1}(X_t) = (X_t - c_t)/\Delta_t$. Hence define $\Theta_t^* = (X_t^* - c_t)/\Delta_t$. This definition gives $\Theta^*_t|Z^*_{\leq t}=z_{\leq t}\eqd \Theta_t|Z_{\leq t}=z_{\leq t}$. Moreover, by using the fact that $X_t^*=m_t(\Theta^*_t)$,
\begin{align*}
\E{\Theta^*_{t+1}|\Theta^*_t, Z^*_{\leq t}}&=\E{\Theta^*_{t+1}|X^*_t, Z_{\leq t}} = X^*_t=m_t(\Theta^*_t).
\end{align*}
This completes the proof.
%  and that \eqref{ap.2}{\color{red}, along with the fact that $m_t(\cdot)$ is one-to-one,} implies
%  \[
%  \E{\Theta_{t+1}^* \mid \Theta_t^*, Z_{\le t}^*}={\color{red}m_t(\Theta_t^*)}, \qquad \text{a.s.},
%  \]
%  which finishes the proof in this case. {\color{red}When $m_t$'s are constant functions}, clearly we must have $m_t(\Theta_t) = \E{\Theta_{t+1}|Z_{\leq t}}$. Thus we may consider the case
% \[
% \Theta_{t+1}^* \perp \Theta_{t}^* \mid Z_{\le t},
% \]
% which is always well-defined.
% Then, we have item {\it ii} implies item {\it i}. This concludes the proof.
\end{proof}

\begin{theorem}[(i) of Theorem 1.4 of \citet{leskela2017condconvex}]\label{condcoupling} For any real-valued random vectors $X$ and $Y$ with finite first moments and any random element $Z$, $\left[X|Z\right] \preceq_{\mathrm{cx}}\left[Y|Z\right]$ if and only if there exists a coupling $(X^*, Y^*, Z^*)$ such that $(X, Z)\stackrel{(d)}{=}(X^*, Z^*)$, $(Y, Z)\stackrel{(d)}{=}(Y^*, Z^*)$ and $\E{Y^*|X^*, Z^*}=X^*$, a.s.
\end{theorem}

\end{document}